\documentclass[letterpaper, onecolumn, 11pt, accepted=2024-05-01]{quantumarticle}
\pdfoutput=1

\usepackage{graphicx,color,amsmath,amsfonts,enumerate,amsthm,amssymb,mathtools,enumitem,thmtools,hyperref,subfigure,mathdots,enumitem,centernot,bm,soul,bbm,stmaryrd}
\usepackage[capitalise, noabbrev]{cleveref}
\usepackage{multirow}

\usepackage{tikz,pgfplots}
\pgfplotsset{compat=1.17}
\usetikzlibrary{quantikz}

\usepackage{algorithm,algpseudocode}
\usepackage{algorithmicx}
\usepackage{mathrsfs}

\let\originalleft\left
\let\originalright\right
\renewcommand{\left}{\mathopen{}\mathclose\bgroup\originalleft}
\renewcommand{\right}{\aftergroup\egroup\originalright}

\def\E{ {\mathbb E} } % using this for expectation

\def\dd{\mathrm{d}}

\def\>{\rangle}
\def\<{\langle}
\newcommand{\abs}[1]{\left| {#1} \right|} 
\newcommand{\ketbra}[2]{\ensuremath{\left|#1\right\rangle\!\left\langle#2\right|}}

\newcommand{\tr}[1]{\mathrm{Tr}\left( #1 \right)}

\newcommand{\norm}[1]{\left\|#1\right\|}

\DeclareMathOperator{\Var}{Var}
\DeclareMathOperator{\Tr}{Tr}
\DeclareMathOperator{\Cov}{Cov}
\DeclareMathOperator{\rank}{rank}
\DeclareMathOperator{\sgn}{sgn}
\newcommand{\rhohat}{\hat{\rho}}

\definecolor{ppblue}{RGB}{46,117,182}
\definecolor{ppred}{RGB}{197, 90, 17}

\newcommand{\hphide}[1]{}
\newcommand{\set}[1]{\left\{#1\right\}}

\newcommand{\inftynorm}[1]{\left\lVert #1 \right\rVert_{\infty}}

\newcommand{\sonenorm}[1]{\left\lVert #1 \right\rVert_1}
\newcommand{\pnorm}[2]{\left\lVert #1 \right\rVert_{#2}}

\newcommand{\bigo}[1]{\mathcal{O}\left(#1\right)}

\newcommand{\bigthetalog}[1]{\tilde{\Theta}\left(#1\right)}

\newcommand{\task}{\mathsf{Shadows}}
\newcommand{\itask}{\mathsf{I}\text{-}\task}
\newcommand{\sym}{\Pi_{\operatorname{sym}}}
\newcommand{\symm}{\operatorname{S}}
\newcommand{\obs}{\mathrm{Obs}}
\newcommand{\BHM}{\operatorname{BHM}}

\theoremstyle{plain}
\newtheorem{thm}{Theorem}
\newtheorem{theorem}[thm]{Theorem}

\newtheorem{lemma}[thm]{Lemma}

\newtheorem{cor}[thm]{Corollary}

\newtheorem{fact}{Fact}
\theoremstyle{definition} 
\newtheorem{defn}[thm]{Definition}

\begin{document}

% -------------------------------------------------------------
% TITLE & ABSTRACT
% -------------------------------------------------------------

\title{Sample-optimal classical shadows for pure states}

\author{Daniel Grier}
\affiliation{Department of Mathematics and Department of Computer Science and Engineering, UC San Diego}
\affiliation{Institute for Quantum Computing, University of Waterloo, Canada}
\author{Hakop Pashayan}
\affiliation{Institute for Quantum Computing, University of Waterloo, Canada}
\affiliation{Department of Combinatorics and Optimization, University of Waterloo, Canada}
\affiliation{Perimeter Institute for Theoretical Physics, Waterloo, Canada}
\affiliation{Dahlem Center for Complex Quantum Systems, Freie Universit\"{a}t Berlin, Germany}
\author{Luke Schaeffer}
\affiliation{Institute for Quantum Computing, University of Waterloo, Canada}
\affiliation{Department of Combinatorics and Optimization, University of Waterloo, Canada}
\affiliation{Joint Center for Quantum Information and Computer Science, University of Maryland, College Park}

\begin{abstract}
We consider the classical shadows task for pure states in the setting of both joint and independent measurements. The task is to measure few copies of an unknown pure state $\rho$ in order to learn a classical description which suffices to later estimate expectation values of observables. Specifically, the goal is to approximate $\mathrm{Tr}(O \rho)$ for any Hermitian observable $O$ to within additive error $\epsilon$ provided $\mathrm{Tr}(O^2)\leq B$ and $\lVert O \rVert = 1$. 
Our main result applies to the joint measurement setting, where we show $\tilde{\Theta}(\sqrt{B}\epsilon^{-1} + \epsilon^{-2})$ samples of $\rho$ are necessary and sufficient to succeed with high probability. The upper bound is a quadratic improvement on the previous best sample complexity known for this problem. For the lower bound, we see that the bottleneck is not how fast we can learn the state but rather how much any classical description of $\rho$ can be compressed for observable estimation.
In the independent measurement setting, we show that $\mathcal O(\sqrt{Bd} \epsilon^{-1} + \epsilon^{-2})$ samples suffice. Notably, this implies that the random Clifford measurements algorithm of Huang, Kueng, and Preskill, which is sample-optimal for mixed states, is not optimal for pure states. Interestingly, our result also uses the same random Clifford measurements but employs a different estimator.
\end{abstract}

\maketitle

% -------------------------------------------------------------
% SEC. I - INTRODUCTION
% -------------------------------------------------------------

\section{Introduction}
\label{sec:intro}

How many copies of an unknown state are required to construct a classical description of the state?  The answer to this question will depend on several details:  what constitutes an accurate description; what is already known about the state; and what restrictions are placed on the measurements of the state. Given the fundamental importance of this question, there has been significant prior work in bounding the number of samples of the states required to perform this learning task in a variety of contexts. 

The most well-known setting is called \emph{quantum state tomography}, where the goal is to learn enough about the state to be able to completely reconstruct it---precisely, estimate the unknown $d$-dimensional quantum state to accuracy $\epsilon$ in the Schatten $1$-norm. Tight upper and lower bounds for the number of copies required for this task are known: $\tilde{\Theta}(\epsilon^{-2} d^3)$ copies of the state are needed with independent measurements~\cite{Haah2016sample}, and $\tilde{\Theta}(\epsilon^{-2} d^2)$ copies are needed when the unknown states can be simultaneously measured in a large joint measurement~\cite{ODonnel2016efficient}. Independent measurements are easier to experimentally implement, while the joint measurements explore what is possible with respect to the fundamental limits of quantum mechanics. A key takeaway in the joint measurement setting is that the algorithm for the upper bound is achieving what would na\"{i}vely be the best possible result, given that a $d$-dimensional state has $\Theta(d^2)$-many independent parameters, and $\Theta(\epsilon^{-2})$ samples are necessary to estimate any one parameter.

In some sense, the requirements of the quantum tomography question are quite rigid. For many applications, only some properties of the unknown state are important. Can we get away with fewer samples if we relax our notion of approximation? In particular, what if we only wish to learn the expected values of certain Hermitian observables? Aaronson gave a somewhat surprising answer to this question in a joint measurement setting called \emph{shadow tomography}~\cite{Aaronson2018shadow}: given $M$ bounded observables 
($\set{O_i}_{i=1}^M$, $\norm{O_i} \le 1$),
estimate $\tr{O_i \rho}$ to within $\epsilon$ additive error.\footnote{Aaronson actually stipulates that each observable is positive semi-definite matrix $E_i$ so that $\{E_i, I-E_i\}$ is a 2-outcome POVM.  We note that this is equivalent to the task of estimating expectation values of the (bounded) Hermitian observables $O_i$ via the mapping $E_i=(O_i+I)/2$.} In this setting, Aaronson showed that only $\tilde{\mathcal O}(\epsilon^{-4} \log^4 M \log d)$ samples of the state are needed. Subsequent work by B\u{a}descu and O'Donnell \cite{buadescu2021improved} improved this to $\tilde{\mathcal O}(\epsilon^{-4} \log^2 M \log d)$, but there are still no matching lower bounds for this setting. That is, we do not know if we are extracting as much information about the unknown state as we can. In the independent measurement setting, $\tilde{\Theta}(\min\{M, d\}/\epsilon^2)$ samples are necessary and sufficient \cite{chen2022exponential}.

One subtlety concerning these observable estimation tasks is whether or not the measurements are allowed to depend on the specific observables $O_i$. In shadow tomography, the measurements \emph{can} depend on the observables, but an increasingly popular setting (inspired by the work of Huang, Kueng, and Preskill \cite{Huang2020}) is one in which the observables $O_i$ are unknown at the time of measurement. That is, the measurements must produce a classical description (called the \emph{classical shadow}) from which the observable expected values can later be calculated. In their randomized Clifford measurement scheme, Huang, Kueng, and Preskill consider the independent measurement setting and show that $\Theta(B \epsilon^{-2} \log M)$ copies of the unknown state are both necessary and sufficient provided that $\tr{O_i^2}\leq B$ for all $i$ (note that $\tr{O_i^2}\leq d$).

Consider now how the classical shadows setting compares to the quantum state tomography setting with regard to the type of measurements allowed.  In the quantum state tomography setting, we know that joint measurements allow us to extract more information from the state, yielding estimates of the unknown state with provably fewer samples than those required with independent measurements.  In the classical shadows setting, however, it is not known how the type of measurement affects the number of samples required.  Concretely, is it possible to perform the classical shadows task with fewer samples if we switch to joint measurements?  We answer this question affirmatively in the setting of \emph{pure} states.

Formally, we show the following: $\mathcal O((\sqrt{B}\epsilon^{-1}+\epsilon^{-2})\log M)$ samples of the unknown pure state are sufficient for performing the classical shadows task with constant probability of failure. Compared to~\cite{Huang2020}, this achieves almost a square root reduction in sample complexity.

Remarkably, in analogy with the quantum state tomography setting, our joint measurement procedure is in some sense extracting the maximum amount of information possible.  To see this, consider a simple setting in which $B = d$, $\epsilon$ is constant, and we only wish to estimate a single observable. Our algorithm uses $\mathcal O(\sqrt d)$ samples. However, Gosset and Smolin~\cite{gosset_smolin2019} show that even if you are given the state as an explicit density matrix, you cannot compress your description of the state down to fewer than $\Omega(\sqrt{d})$-many bits of information in order to estimate arbitrary observable expectation values. Notice, however, that to successfully execute the classical shadows task, one would first need to learn such a compressed description through measurement of the unknown state. A priori, the number of measurements required to do this could be much higher than the size of this compressed description.
The fact that we find a matching upper bound implies that accessing the relevant information contained in the state is not the significant bottleneck.

We show that a similar phenomena exists for arbitrary parameters $B$ and $\epsilon$. Namely, we refine the Gosset-Smolin lower bound for compression to $\Omega(\sqrt{B}\epsilon^{-1})$-many bits, which ultimately allows us to show that $\tilde{\Omega}(\sqrt{B}\epsilon^{-1}+\epsilon^{-2})$ samples of the state are required for the classical shadows task. Therefore, our joint-measurement algorithm above is sample-optimal (at least for a single observable and up to log factors).

Finally, we address the classical shadows question with pure states and independent measurements. We show that $\mathcal O((\frac{\sqrt{Bd}}{\epsilon} + \frac{1}{\epsilon^2})\log M)$ copies of the state suffice. It's worth noticing that in certain parameter regimes, this upper bound is smaller than $\Theta(B \epsilon^{-2} \log M)$. In other words, our algorithm uses fewer samples than the classical shadows algorithm of Huang, Kueng, and Preskill which was designed for general mixed states. Indeed, their lower bound methods require the underlying state to be mixed.

\subsection{The classical shadows task}
\label{sec:classical_shadows_task}
We consider the classical shadows task introduced by Huang, Kueng, and Preskill \cite{Huang2020}: given several copies of an unknown quantum state, produce a classical description of the state that is sufficiently representative to permit the reliable and accurate estimation of expectation values of some number of observables chosen from a broad class.

To formalise the task, let's begin with the class of observables we will use: 
\begin{defn}
For any $B\in (0,d]$, let
$$\obs(B):=\set{O\in \mathbb C^{d \times d}~|~O=O^{\dagger},\inftynorm{O} = 1, \Tr(O^2)\leq B}.$$
\end{defn}
In summary, these observables have been scaled/normalized so that $\inftynorm{O} = 1$ and have a bound of $B$ on their squared Frobenius norm $\Tr(O^2)$. The latter condition is due to the fact that $\Tr(O^2)$ is typically the dominant term in the sample complexity. We could also reasonably upper bound it by the rank of the observable since $\Tr(O^2) \leq \rank{O} \leq d$.

We remark that $\pnorm{O}{2} = \sqrt{\Tr(O^2)}$ and $\inftynorm{O}$ are examples of Schatten $p$-norms where $p = 2$ and $p = \infty$ respectively, but defined in general as $\pnorm{A}{p} := \Tr(|A|^{p})^{1/p}$ for $p \in [1, \infty)$. We will also use the Schatten $1$-norm. Going forward, we write $\pnorm{O}{1}$ for the $1$-norm, $\norm{O}$ for the infinity norm, and prefer $\Tr(O^2)$ over $\pnorm{O}{2}^2$. 

\begin{defn}[Classical Shadows Task]\label{defn:classical_shadows_task}
    The Classical Shadows Task consists of two separate phases---a measurement phase and an observable estimation phase---which are completed by two separate (randomized) algorithms, $\mathcal A_{\mathrm{meas}}$ and $\mathcal A_{\mathrm{est}}$, respectively. In addition to the inputs below, each algorithm also depends on the four parameters $s$, $B$, $\epsilon$, and $\delta$:
    \begin{itemize}
    \item[] \textbf{Measurement:} $\mathcal A_{\mathrm{meas}} \colon \rho^{\otimes s} \to \{0,1\}^*$ \\
        \noindent \textit{Input:}  $s$ copies of a state $\rho \in \mathbb C^{d \times d}$. \\
        \noindent \textit{Output:} A bit string called the \emph{classical shadow}. 
    \item[] \textbf{Estimation:} $\mathcal A_{\mathrm{est}} \colon \obs(B) \times \{0,1\}^* \to \mathbb R$\\
        \noindent \textit{Input:} Observable $O\in \obs(B)$ and a classical shadow. \\
        \noindent \textit{Output:} Estimate $E \in \mathbb R$.
    \end{itemize}
    \noindent It's worth emphasizing that the input to the measurement algorithm is quantum (the state $\rho^{\otimes s}$) and the output is classical (the classical shadow). This output is computed from measuring the input state with some POVM (with arbitrary post-processing). 
    
    We say that $\mathcal A_{\mathrm{meas}}$ and $\mathcal A_{\mathrm{est}}$ constitute a valid protocol for the classical shadows task if their estimate for the expectation of the observable $E := \mathcal A_{\mathrm{est}}(O, \mathcal A_{\mathrm{meas}}(\rho^{\otimes s}))$ is such that 
    \begin{align}\label{eq:close_in_expectation}
    \abs{\Tr(O \rho) - E} < \epsilon
    \end{align}
    with probability at least $1-\delta$ over the randomness of $\mathcal A_{\mathrm{meas}}$ and $\mathcal A_{\mathrm{est}}$. 
\end{defn}

Some may find it useful to think about the classical shadows task as a one-way communication protocol where one party (let's call her Melanie) is given copies of an unknown state and another party (say, Esteban) is given an observable. Melanie doesn't know Esteban's observable, and Esteban cannot send hints because we are assuming one way communication from Melanie to Esteban, so there is only one course of action: Melanie must measure her unknown state and send (over a classical channel) a description of the state from which Esteban can estimate the expected value of his given observable. 

Throughout this paper, we will focus on the classical shadows task with unknown \emph{pure} states. This motivates the following definitions:

\begin{defn}[Sample Complexity of the Classical Shadows Task]
    Let $\task(B,\epsilon,\delta)$ to be the minimum number of samples $s$ required to successfully carry out the classical shadows task on pure states with the set of observables $\obs(B)$, to accuracy $\epsilon$, and failure probability at most $\delta$.
    
    Sometimes we will omit $\delta$ and write $\task(B,\epsilon)$ to denote the minimum number of samples to achieve these tasks with some constant probability of failure, say, $0.001$. 
\end{defn}

\begin{defn}[Classical Shadows with Independent Measurements]
 Let $\itask(B,\epsilon,\delta)$ be the sample complexity for the classical shadows task with pure states when the measurement algorithm can only make \emph{independent} measurements on the input state---that is, the measurement POVM is the tensor product of POVMs on single copies of the state. These POVMs do not have to be identical, but the entire state must be measured at the same time, or in other words, the output from a measurement on one copy of the state cannot influence the measurement on another.
\end{defn}

We note that there are many possible variants for the sample complexity of the classical shadows task that we haven't given individual names. Most notably are the settings where the unknown states are mixed states (rather than pure) and/or the measurements are allowed to be adaptive (while still acting on single copies of the state).

\subsection{Summary of results}
Our main result is to prove matching upper and lower bounds on the sample complexity of performing the classical shadows task with respect to joint measurements and pure states.
\begin{thm}
\label{thm:ps_jm_ubANDlb}
$\task(B,\epsilon) = \bigthetalog{\frac{\sqrt{B}}{\epsilon} + \frac{1}{\epsilon^2}}$ provided $B \le \epsilon d$.
\end{thm}

Notice that \Cref{thm:ps_jm_ubANDlb} consists of separate upper and lower bound results (for constant $\delta$). These match up to logarithmic factors in $B$ and $\epsilon^{-1}$, and the technical relationship between $B$, $\epsilon$, and $d$ is only required for the lower bound. In \Cref{sec:upper_bound}, we will prove the upper bound, where we will also show that the dependence on the failure probability $\delta$ goes as $\log(1/\delta)$. We note that this dependence on $\delta$ implies that there are efficient protocols for the calculation of several observables simultaneously---that is, if the classical shadows task fails with probability at most $\delta$ on a single observable, then it fails with probability at most $M\delta$ on one or more out of $M$ observables by the union bound. In \Cref{sec:lower_bound}, we will prove the lower bound where only the $\epsilon^{-2}$ term will scale with $\log(1/\delta)$. 

We also prove an upper bound on the sample complexity of performing the classical shadows task with respect to independent measurements and \emph{pure} states. Our upper bound can be compared to the matching upper and lower bound of Huang, Kueng, and Preskill \cite{Huang2020} which applies to independent measurements and \emph{general} states. In certain parameter regimes, our upper bound achieves a \emph{smaller} sample complexity than the lower bound in \cite{Huang2020} which implies that in the independent measurement setting, the classical shadows task has smaller sample complexity for pure states.
\begin{thm}
\label{thm:ps_im_ub_summary}
    For all $\epsilon, \delta > 0$,
    $$\itask(B, \epsilon, \delta) = \bigo{\min \left\{ \frac{B}{\epsilon^2}, \frac{\sqrt{Bd}}{\epsilon} + \frac{1}{\epsilon^2} \right \} \log(\delta^{-1})}.$$
\end{thm}
We discuss and prove \Cref{thm:ps_im_ub_summary} in \Cref{sec:independent_measurements}.

Finally, we note that in all of our algorithms, the estimator $\hat{\rho}$ we use for the unknown state $\rho$ is not itself a proper state. In \Cref{sec:proper_learning}, we show that this is a necessary price for the favorable sample complexity enjoyed by classical shadows schemes. Informally, we show that even for observables in $\obs(1)$, learning an estimate $\hat{\rho}$  that is a proper state to sufficient accuracy to solve the classical shadows task via the formula $\tr{O\hat{\rho}}$, requires a sample complexity that scales linearly in $d$, the dimension of the unknown state.

% -------------------------------------------------------------
% SEC. II - PRELIMINARIES
% -------------------------------------------------------------

\section{Preliminaries}
Here we cover key background material related to Haar random states, their moments, and the symmetric subspace. Throughout, we're working with \emph{qudits} of dimension $d \ge 2$ unless otherwise specified. Since the unitary group $\mathrm{U}(d)$ acts on the Hilbert space of dimension $d$, it has a corresponding \emph{Haar measure} which is invariant under the action of the group. \emph{Haar random states} sampled proportional to this measure are ubiquitous in quantum information, and essential to define our measurement in Section~\ref{sec:upper_bound}. 

To perform the necessary calculations on Haar random states, we need to discuss their moments, and some ancillary concepts.
\begin{defn}
    For integer $k \geq 1$, \emph{$k$-th moment} of an ensemble $\mathcal{E}$ of quantum states is 
    $$
    \E_{\ket{\psi} \sim \mathcal{E}}[ \ketbra{\psi}{\psi}^{\otimes k}].
    $$
    An ensemble $\mathcal{E}$ is a \emph{(state) $t$-design} if the moments $1 \leq k \leq t$ are identical to those of the Haar distribution (see Lemma~\ref{lem:rep_theory}). 
\end{defn}

\begin{defn}[permutation operator]
    Given a permutation $\pi \in \mathrm{S}_s$ (for $s \geq 1$), define a permutation operator $W_{\pi} \in \mathbb C^{d^{s} \times d^{s}}$ such that 
    $$
    W_{\pi} \ket{x_1} \cdots \ket{x_s} = \ket{x_{\pi^{-1}(1)}} \cdots \ket{x_{\pi^{-1}(s)}},
    $$
    and extend by linearity. 
    That is, $W_{\pi}$ acts on $(\mathbb C^{d})^{\otimes s}$ by permuting the qudits, sending the qudit in position $i$ to position $\pi(i)$. 
\end{defn}

\begin{defn}[symmetric subspace]
    The symmetric subspace of an $s$-qudit system $(\mathbb C^{d})^{\otimes s}$ is the subspace invariant under $W_{\pi}$ for all $\pi \in \mathrm{S}_s$. We use $\kappa_s$ to denote its dimension and define $\sym^{(s)}$ to be the projector onto it (notationally omitting the dependence on $d$, the dimension of the qudit). 
\end{defn}

We have two characterizations of the symmetric subspace.
\begin{fact}
\label{symmetric_subspace_size}
    For all $s$,
    $\sym^{(s)} = \frac{1}{s!} \sum_{\pi \in \mathrm{S}_s} W_{\pi}$, and 
    $\kappa_s = \binom{s+d-1}{d-1}$.
\end{fact}

The integral of $\ketbra{\psi}{\psi}$ over the Haar measure is known from, e.g., \cite{scott2006tight}.
\begin{lemma}
\label{lem:rep_theory}
$$
\kappa_s \int_{\psi} \left( \ketbra{\psi}{\psi} \right)^{\otimes s} \mathrm{d}\psi = \sym^{(s)} = \frac{1}{s!} \sum_{\pi \in \mathrm{S}_{s}} W_{\pi}
$$  
where $\sym^{(s)}$ is the projector onto the symmetric subspace and $W_{\pi}$ is the operator that permutes $s$ qudits by an $s$-element permutation $\pi$.  
\end{lemma}

We will often need to compute the (partial) trace of $(A_1 \otimes A_2 \otimes \cdots \otimes A_s) W_{\pi}$ for some linear operators $A_1, \dots, A_s \in \mathbb C^{d \times d}$. It turns out that there is an extremely useful tensor network based pictorial representation that simplifies these calculations. Let us give a brief introduction to those techniques, though readers may also find more thorough treatments useful \cite{gross2015partial, roberts2017chaos}.

To start, we draw a single $d$ dimensional linear operator $A = \sum_{i,j \in [d]} a_{i,j} \ket{i}\bra{j}$ as a tensor block with a leg for the input and output indices for $A$:
\begin{center}
\begin{tikzpicture}[baseline=(current bounding box.center)]
    \node(A) {$A$};
	\draw (A.north) edge node[left,very near end] {$i$} ++(0,1em);
	\draw (A.south) edge node[right,very near end] {$j$} ++(0,-1em);
\end{tikzpicture}
\end{center}
Suppose we have another tensor $B = \sum_{i,j \in [d]} b_{i,j} \ket{i}\bra{j}$. We express composition, tensor product, and trace as the following tensor networks:

\def \loopwidth {1em}
\def \colspacing {2em}
\def \rowspacing {1.5em}
\setlength\tabcolsep{20pt}
\newsavebox{\composition}
\newsavebox{\tensorproduct}
\newsavebox{\traceytrace}
\savebox{\composition}{
\begin{tikzpicture}[baseline=(current bounding box.center)]
	\matrix (m) [matrix of math nodes, row sep=1em, column sep=1.5em, text height=0.8em,text depth=.2em] {
		A \\
		B \\
	};
	\draw (m-1-1) -- (m-2-1);
	\draw (m-2-1.south) -- ++(0,-1em);
	\draw (m-1-1.north) -- ++(0,1em);
\end{tikzpicture}
}%
\savebox{\tensorproduct}{
\begin{tikzpicture}[baseline=(current bounding box.center)]
	\matrix (m) [matrix of math nodes, row sep=\rowspacing, column sep=1em, text height=0.8em,text depth=.2em] {
		& \\
		A & B \\
		& \\
	};
	\draw (m-1-1) -- (m-2-1) -- (m-3-1);
	\draw (m-1-2) -- (m-2-2) -- (m-3-2);
\end{tikzpicture}
}%
\savebox{\traceytrace}{
\begin{tikzpicture}[baseline=(current bounding box.center)]
	\matrix (m) [matrix of math nodes, row sep=\rowspacing, column sep=\colspacing, text height=0.8em,text depth=.2em] {
		A \\
	};
	\draw[rounded corners] (m-1-1.south) -- ++(0,-1em) -- ++(\loopwidth,0) |- ($(m-1-1.north) + (0,1em)$) -- ++ (m-1-1);
\end{tikzpicture}
}
\begin{center}
\begin{tabular}{c c c}
Composition ($AB$) & Tensor Product ($A \otimes B$) & Trace ($\Tr(A)$) \\ \hline
\usebox\composition & \usebox\tensorproduct & \usebox\traceytrace
\end{tabular}
\end{center}
The reason the tensor network picture is particularly nice for dealing with traces of $W_{\pi}$ terms is because each $W_{\pi}$ term is simply a permutation of wires in the tensor network picture. For example, 
for a simple cyclic permutation, we have
\[
\begin{tikzpicture}[baseline=(current bounding box.center)]
    \node[draw,rectangle,minimum width=4em] (W) {$W_{(1\,2\,3)}$};
    \draw (W.north) ++ (1.5em,0) -- ++(0,1em);
    \draw (W.north) ++ (0em,0) -- ++(0,1em);
    \draw (W.north) ++ (-1.5em,0) -- ++(0,1em);
    \draw (W.south) ++ (1.5em,0) -- ++(0,-1em);
    \draw (W.south) ++ (0em,0) -- ++(0,-1em);
    \draw (W.south) ++ (-1.5em,0) -- ++(0,-1em);
    \end{tikzpicture}
\,\,\, =
\begin{tikzpicture}[baseline=(current bounding box.center)]
	\matrix (m) [matrix of math nodes, row sep=\rowspacing, column sep=1em, text height=0.8em,text depth=.2em] {
		\phantom{A} & \phantom{B} & \phantom{C} \\
	};
	\draw[rounded corners] ($(m-1-1.north)+(0,1em)$) -- (m-1-1.north) -- (m-1-3.south) -- ++(0,-1em);
	\draw[rounded corners] ($(m-1-2.north)+(0,1em)$) -- (m-1-2.north) -- (m-1-1.south) -- ++(0,-1em);
	\draw[rounded corners] ($(m-1-3.north)+(0,1em)$) -- (m-1-3.north) -- (m-1-2.south) -- ++(0,-1em);
\end{tikzpicture}.
\]

The key feature of tensor networks is that only the topology of the network matters, so we can simplify tensor networks just by moving the elements around. For example, consider a common partial trace  that will arise in this paper: $\Tr_1((A \otimes B) W_{(1 \, 2)})$. Drawing the tensor network, we get 
\[
\Tr_1((A \otimes B) W_{(1\,2)}) \quad = \quad 
\begin{tikzpicture}[baseline=(current bounding box.center)]
	\matrix (m) [matrix of math nodes, row sep=\rowspacing, column sep=\colspacing, text height=0.8em,text depth=.2em] {
		A & B \\
		& \\
	};
	
	\draw (m-1-2.north) -- ++(0,1em);
	\draw[rounded corners=0.5em] (m-1-1) -- (m-2-1.north) -- (m-2-2.south) -- ++(0,-1em);
	\draw[rounded corners=0.5em] (m-1-2) -- (m-2-2.north) -- (m-2-1.south) -- ++(0,-1em) -- ++(-\loopwidth,0) |- ($(m-1-1.north) + (0,1em)$) -- ++ (m-1-1);
\end{tikzpicture}
\,=\,\,
\begin{tikzpicture}[baseline=(current bounding box.center)]
	\matrix (m) [matrix of math nodes, row sep=\rowspacing, column sep=\colspacing, text height=0.8em,text depth=.2em] {
		\phantom{B} & B \\
		\phantom{A} & A \\
	};
	
	\draw (m-2-2.south) -- ++(0,-1em);
	\draw (m-1-2.north) -- ++(0,1em);
	\draw[rounded corners] (m-1-2.south) -- (m-2-1.north) -- (m-2-1.south) -- ++(0,-1em) -- ++ (-1em,0) |- ($(m-1-1.north)+(0,1em)$) -- (m-1-1.south) -- (m-2-2.north);
\end{tikzpicture}
\,=\,
\begin{tikzpicture}[baseline=(current bounding box.center)]
	\matrix (m) [matrix of math nodes, row sep=\rowspacing, column sep=\colspacing, text height=0.8em,text depth=.2em] {
		B \\
		A \\
	};
	
	\draw (m-1-1) -- (m-2-1);
	\draw (m-2-1.south) -- ++(0,-1em);
	\draw (m-1-1.north) -- ++(0,1em);
\end{tikzpicture}
= \quad BA
\]
where we can push the $B$ tensor through the SWAP and around the trace loop to see that it is composed with $A$. In other words, we have just shown the identity $\Tr_1((A \otimes B) W_{(1 \, 2)}) = BA$.

As a generalization, we have the following useful fact:
\begin{fact}
    \label{fact:two}
    Let $n \geq 1$ and $\pi = (1 \, 2 \, \cdots \, n) \in \symm_n$. For any $A_1, \ldots, A_n$, we have
    $$
    \Tr_{-1}(W_{\pi}(A_1 \otimes A_2 \otimes \cdots \otimes A_n)) = A_n A_{n-1} \cdots A_1,
    $$
    where $\Tr_{-1}$ indicates the partial trace of all but the first qudit. Thus, $\Tr(W_{\pi}(A_1 \otimes A_2 \otimes \cdots \otimes A_n)) = \Tr(A_n A_{n-1} \cdots A_1)$.
\end{fact}
\begin{proof}
The fact is best seen with a small example. When $n = 5$, for instance, the tensor network diagram is
\[
 	\begin{tikzpicture}[baseline=(current bounding box.center)]
		\matrix (m) [matrix of math nodes, row sep=\rowspacing, column sep=\colspacing, text height=0.8em,text depth=.2em] {
			& & & & \\
			A_1 & A_2 & A_3 & A_4 & A_5 \\
		};
		
		\draw[rounded corners=0.5em] ($(m-1-1.north) + (0,1em)$) -- (m-1-1.north) -- (m-1-5.south) -- (m-2-5.north);
		
		\draw[rounded corners=0.5em] (m-2-2.south) -- ++(0,-1em) -- ++(\loopwidth,0) |- ($(m-1-2.north) + (0,1em)$) -- (m-1-2.north) -- (m-1-1.south) -- (m-2-1.north);
		
		\draw[rounded corners=0.5em] (m-2-3.south) -- ++(0,-1em) -- ++(\loopwidth,0) |- ($(m-1-3.north) + (0,1em)$) -- (m-1-3.north) -- (m-1-2.south) -- (m-2-2.north);
		
		\draw[rounded corners=0.5em] (m-2-4.south) -- ++(0,-1em) -- ++(\loopwidth,0) |- ($(m-1-4.north) + (0,1em)$) -- (m-1-4.north) -- (m-1-3.south) -- (m-2-3.north);

		\draw[rounded corners=0.5em] (m-2-5.south) -- ++(0,-1em) -- ++(\loopwidth,0) |- ($(m-1-5.north) + (0,1em)$) -- (m-1-5.north) -- (m-1-4.south) -- (m-2-4.north);
		
		\draw (m-2-1.south) -- ++ (0,-1em);
	\end{tikzpicture}
\]
from which the identity follows.
\end{proof}

% -------------------------------------------------------------
% SECTION - JOINT MEASUREMENT UPPER BOUND
% -------------------------------------------------------------

\section{Joint Measurement Upper Bound}
\label{sec:upper_bound}

The goal of this section is to prove the following upper bound for the sample complexity of classical shadows for pure states and joint measurements.
\begin{thm}
\label{thm:ps_jm_ub}
$$\task(B,\epsilon,\delta) = \bigo{\left(\frac{\sqrt{B}}{\epsilon}+ \frac{1}{\epsilon^2}\right)\log \frac{1}{\delta}}.$$ 
\end{thm}

The proof of \Cref{thm:ps_jm_ub} is constructive; given $B$, $\epsilon$, $\delta$ and $d$, we specify the number of samples and a pair of algorithms $\mathcal A_{\mathrm{meas}}$ and $\mathcal A_{\mathrm{est}}$ that solve the classical shadows task with that many samples.

In brief, the construction is as follows. We give a measurement $\mathcal{M}_{s}$ on $s$ copies of $\rho$, where the outcome of the measurement is a classical description of a pure state $\Psi$. We apply an affine transformation to the outcome $\Psi$ to produce an unbiased ``shadow'' estimator: a unital Hermitian matrix $\rhohat$ such that $\E[\rhohat] = \rho$. 
Increasing the number of samples, $s$,  suppresses the additive error $\epsilon$ and the failure probability $\delta$ by a factor of $s^{-\bigo{1}}$. To improve this to an inverse exponential suppression in the failure probability,
we repeat the entire procedure $k = \mathcal O(\log(\delta^{-1}))$ times and take the median of the batch estimates akin to the \emph{median of means} method~\cite{lugosi2019mean,lerasle2019lecture,Huang2020}. A pseudocode description is given in Algorithms~\ref{alg:ps_jm_meas} and~\ref{alg:ps_jm_est}.

%XXXXXXXXXXXXXXXXXXXXXXXXXXXXXXXXXXXXXXXXXXXXXXXXXXXXXXXXXXXXX
%Measure Alg
%XXXXXXXXXXXXXXXXXXXXXXXXXXXXXXXXXXXXXXXXXXXXXXXXXXXXXXXXXXXXX
\begin{algorithm}[H]
\caption{Algorithm for $\mathcal{A}_{\mathrm{meas}}$ of  Theorem~\ref{thm:ps_jm_ub}}\label{alg:ps_jm_meas}
\begin{algorithmic}[1]
\algrenewcommand\algorithmicrequire{\textbf{Input:}}
\algrenewcommand\algorithmicensure{\textbf{Output:}}
\Require Quantum state $\rho^{\otimes N}$, $B$, $\epsilon$, $\delta$, $d$.
\Ensure Classical shadow $\set{\hat{\rho}^{(i)}}_{i \in [k]}$.
\Statex
\State{$s \gets \mathcal O(\sqrt{B} \epsilon^{-1} + \epsilon^{-2})$} \Comment{Samples per batch}
\State{$k \gets \lfloor{N/s}\rfloor$} \Comment{Number of batches}
\For{each batch $i = 1, \ldots, k$}
    \State{$\psi_i \gets $ Measure new batch of $\rho^{\otimes s}$ with $\mathcal{M}_s$}
    \State{$\rhohat^{(i)} \gets \frac{(d+s)\psi_i - I}{s}$}
\EndFor
\State{\Return $\set{\hat{\rho}^{(i)}}_{i \in [k]}$}
\end{algorithmic}
\end{algorithm}
%XXXXXXXXXXXXXXXXXXXXXXXXXXXXXXXXXXXXXXXXXXXXXXXXXXXXXXXXXXXXXXXXX

%XXXXXXXXXXXXXXXXXXXXXXXXXXXXXXXXXXXXXXXXXXXXXXXXXXXXXXXXXXXXX
%Estimate Alg
%XXXXXXXXXXXXXXXXXXXXXXXXXXXXXXXXXXXXXXXXXXXXXXXXXXXXXXXXXXXXX
\begin{algorithm}[H]
\caption{Algorithm for $\mathcal{A}_{\mathrm{est}}$ of Theorem~\ref{thm:ps_jm_ub} and \Cref{thm:ps_im_ub}}\label{alg:ps_jm_est}
\begin{algorithmic}[1]
\algrenewcommand\algorithmicrequire{\textbf{Input:}}
\algrenewcommand\algorithmicensure{\textbf{Output:}}
\Require Classical shadow $\set{\hat{\rho}^{(i)}}_{i \in [k]}$ and observable $O$. 
\State{$E \gets \operatorname{median}(\Tr(O \rhohat^{(1)}), \ldots, \Tr(O \rhohat^{(k)}))$}
\State{\Return $E$}
\end{algorithmic}
\end{algorithm}
%XXXXXXXXXXXXXXXXXXXXXXXXXXXXXXXXXXXXXXXXXXXXXXXXXXXXXXXXXXXXXXXXX

We now define our measurement $\mathcal{M}_s$. 
\begin{defn}
\label{defn:measurement}
The \emph{standard symmetric joint measurement} is a measurement on $s$ qu\emph{d}its. It is defined by the POVM $\mathcal{M}_{s} = \{ A_{\psi} \}_{\psi} \cup \{ I - \sym^{(s)} \}$ with elements
$$
A_{\psi} := \kappa_s \ketbra{\psi}{\psi}^{\otimes s} \mathrm{d} \psi,
$$
for all $d$-dimensional pure states, proportional to the Haar measure, plus a ``fail'' outcome $I - \sym^{(s)}$ for non-symmetric states.
\end{defn}
We will be interested in the setting were $\rho$ is pure\footnote{This is the first time we use the purity of $\rho$ in our analysis, but certainly not the last.} and therefore $\rho^{\otimes s}$ is in the symmetric subspace, so we will never see the ``fail'' outcome---it exists solely to make the POVM sum/integrate to $I$.

One might be concerned that the standard symmetric joint measurement is constructed from the Haar measure, resulting in a continuum of outcomes. This is technically inconsistent with \cref{defn:classical_shadows_task} where the measurement must output a \emph{finite} length bit string. However, it will turn out that our analysis (c.f. \cref{thm:ps_jm_ub}) only requires the states that appear in the POVM to form an $(s+2)$-design, where $s$ is the number of samples jointly measured. That is, it suffices to replace the continuous POVM $\mathcal M_s$ with a finite POVM $\{\kappa_s p_i \ketbra{\psi_i}{\psi_i}^{\otimes s} \}_i \cup \{ I - \sym^{(s)} \}$ such that $\sum_i p_i \ketbra{\psi_i}{\psi_i}^{\otimes(s+2)} = \int_\psi \ketbra{\psi}{\psi}^{\otimes (s+2)} \mathrm{d} \psi$ where the $p_i \ge 0$ define a finite probability distribution. 

For some perspective, consider the independent measurement setting in which $s=1$. By the above observation, we require the measurement to form a 3-design. Since the set of multi-\emph{qubit} stabilizer states forms a 3-design, we recover the efficient measurement protocol of \cite{Huang2020}. That said, our measurements typically involve many copies of the state, resulting in a large $s$. In such cases, we must use much more complicated constructions of designs (see, e.g., \cite{bajnok1992construction, hayashi2005reexamination, bondarenko2013optimal}). Nevertheless, these constructions result in a finite POVM that can at least in principle be implemented with a projective measurement using $\mathsf{poly}(d, \log(1/\epsilon))$-many ancillas \cite{nielsen2010quantum}.

\subsection{Analysis}
After defining the measurement, the estimator, and how many samples we need, the only remaining technical component is to bound the probability of failure. This ultimately comes down to Chebyshev's inequality:
$$\Pr\left[|\Tr(O \rhohat) - \E[\Tr(O \rhohat)]| \geq \epsilon \right] \leq \frac{\Var(\Tr(O \rhohat))}{\epsilon^2}.$$
Hence, we need to calculate the mean and variance of the random variable $\Tr(O \rhohat)$. To be precise, let $\rho$ be a pure state and suppose we measure $\rho^{\otimes s}$ with the standard symmetric joint measurement $\mathcal{M}_{s}$. Let $\Psi$ be the density matrix random variable for $\ketbra{\psi}{\psi}$, where $\psi$ is the outcome of the measurement. Let's start with the mean:
\begin{lemma}[First moment]
\label{lem:first_moment}
For measurement $\mathcal M_s$ on pure state $\rho^{\otimes s}$, we have
\begin{align*}
    \E[\Psi] &= \frac{I + s\rho}{d+s}.
\end{align*}
\end{lemma}
\begin{proof}
To start, let's express the expectation as a Haar integral using the definition of $\mathcal M_s$:
$$
\E[\Psi] = \int \psi \cdot \Pr[\Psi = \psi] 
= \int \psi \cdot \Tr(A_\psi \rho^{\otimes s})
= \kappa_{s} \int \psi \cdot \Tr( \psi^{\otimes s} \rho^{\otimes s}) \, \dd\psi.
$$
Using the identity $A \Tr(B) = \Tr_2 (A \otimes B)$ for all square matrices $A$ and $B$, we can apply \Cref{lem:rep_theory} to compute the integral above:
$$
\E[\Psi] = \kappa_{s} \int \Tr_{-1} \left( \psi^{\otimes s+1} \cdot (I \otimes \rho^{\otimes s}) \right) \, \dd \psi 
= \frac{\kappa_{s}}{\kappa_{s+1}} \frac{1}{(s+1)!} \sum_{\pi \in \symm_{s+1}} \Tr_{-1} (W_{\pi} (I \otimes \rho^{\otimes s})).
$$
We attack the right hand side by evaluating $\Tr_{-1}(W_{\pi}(I \otimes \rho^{\otimes s})$ for each $\pi$. 
In particular, we will show that
    $$
    \Tr_{-1}( W_{\pi} (I \otimes \rho^{\otimes s})) = \begin{cases}
    I, & \text{if $\pi(1) = 1$,} \\
    \rho, & \text{otherwise.}
    \end{cases}
    $$
    To do this, we take the cycle decomposition of $\pi$ and analyze each cycle separately. Notice that any cycle not involving position $1$ is completely traced out and the cycle operator acts on a tensor power of $\rho$ only, so Fact~\ref{fact:two} says the trace is $\Tr( \rho^{k} ) = \Tr( \rho ) = 1$ (since $\rho$ is pure). Thus, only the cycle through position $1$ matters. If $\pi(1) = 1$, then this cycle is trivial, and the result is $I$. Otherwise, the cycle visits $k \geq 1$ copies of $\rho$, leading to the product $\rho^k = \rho$. 
    
    There are $s!$ permutations which fix $1$ (i.e., $\pi(1) = 1$) and hence $s \cdot s!$ which do not, so we conclude that
    $$
    \E[\Psi] = \frac{\kappa_{s}}{\kappa_{s+1}} \frac{1}{(s+1)!} \sum_{\pi \in \symm_{s+1}} \Tr_{-1} (W_{\pi} (I \otimes \rho^{\otimes s})) = \frac{s! \cdot I + s \cdot s! \cdot \rho}{(d+s) s!} = \frac{I + s \rho}{d+s}. 
    $$
    Notice that $\Tr( \E[\Psi] ) = \frac{\Tr(I) + s \Tr(\rho)}{d+s} = 1 = \E[\Tr(\Psi)]$, as a sanity check.
\end{proof}

We now turn to the variance calculation, which depends on the second moment of the estimator:
\begin{lemma}[Second moment]
    \label{lem:second_moment}
For measurement $\mathcal M_s$ on pure state $\rho^{\otimes s}$, we have
    \begin{align*}
        \E[\Psi \otimes \Psi] &= \frac{2}{(d+s)(d+s+1)} \left( (I + s \rho)^{\otimes 2} - \frac{s(s+1)}{2} (\rho \otimes \rho) \right) \sym^{(2)}
    \end{align*}
\end{lemma}
\begin{proof}
    As in \Cref{lem:first_moment}, we evaluate $\E[\Psi \otimes \Psi]$ as
    \begin{align*}
    \E[\Psi \otimes \Psi] 
    &= \int (\psi \otimes \psi) \cdot \Pr[\Psi = \psi] \\
    &= \int (\psi \otimes \psi) \cdot \kappa_{s} \Tr( \psi^{\otimes s} \rho^{\otimes s}) \, \dd \psi \\
    &= \kappa_{s} \int \Tr_{-1,2} \left( \psi^{\otimes s+2} \cdot (I^{\otimes 2} \otimes \rho^{\otimes s}) \right) \, \dd \psi \\
    &= \frac{\kappa_{s}}{\kappa_{s+2}} \frac{1}{(s+2)!} \sum_{\pi \in \symm_{s+2}} \Tr_{-1,2} (W_{\pi} (I^{\otimes 2} \otimes \rho^{\otimes s})),
    \end{align*}
    where the partial trace $\Tr_{-1,2}$ now preserves the first two qudits. In \Cref{fig:second_moment} and for the special case of $s=1$, we show a complete derivation of how this trace simplifies using the tensor network notation, which may be useful to some readers before proceeding to the more general proof.

    Let us evaluate the sum term-by-term. We divide the permutations into two types: those where $1$ and $2$ appear in separate cycles (type A) and those where $1$ and $2$ appear in the same cycle (type B). Consider the type A permutations first:
    $$
    \Tr_{-1,2}(W_{\pi} (I^{\otimes 2} \otimes \rho^{\otimes s})) = \begin{cases}
    I \otimes I & \text{if $\pi(1) = 1$ and $\pi(2) = 2$} \\
    I \otimes \rho & \text{if $\pi(1) = 1$ and $\pi(2) \neq 2$} \\
    \rho \otimes I & \text{if $\pi(1) \neq 1$ and $\pi(2) = 2$} \\
    \rho \otimes \rho & \text{if $\pi(1) \neq 1$ and $\pi(2) \neq 2$} 
    \end{cases}
    $$
    \begin{figure}[h]
        \centering
        \def\rowspacing{0.3em}
\def\colspacing{.1em}
\def\roundness{0.3em}
\def\overdist{0.5em}
\begin{flalign*}
&\mathbb{E}[\Psi \otimes \Psi] = d \cdot \mathbb{E}_{\psi \sim \text{Haar}} \left[ \psi \otimes \psi  \Tr(\rho \psi) \right]&
\end{flalign*}
\vspace{-1.0cm}
\begin{eqnarray*}
&=& 
d  \cdot \mathbb{E}_{\psi \sim \text{Haar}} \left[
\begin{tikzpicture}[baseline=(current bounding box.center)]
	\matrix (m) [matrix of math nodes, row sep=\rowspacing, column sep=\colspacing, text height=0.8em,text depth=.2em] {
		& & \rho \\
		\psi & \psi & \psi \\
	};
	
	\draw[rounded corners=\roundness] ($(m-1-1.north) + (0,\overdist)$) -- (m-2-1);
	\draw (m-2-1.south) -- ++(0,-\overdist);
	\draw[rounded corners=\roundness] ($(m-1-2.north) + (0,\overdist)$) -- (m-2-2);
	\draw (m-2-2.south) -- ++(0,-\overdist);
	\draw (m-1-3) -- (m-2-3);
	\draw[rounded corners=\roundness] (m-2-3.south) -- ++(0,-\overdist) -- ++(\loopwidth,0) |-  ($(m-1-3.north) + (0,\overdist)$) -- (m-1-3);
\end{tikzpicture}\,\,\,
\right] \\
&=& \frac{1}{(d+1)(d+2)} \sum_{\pi \in \mathrm{S}_3} \left[
\begin{tikzpicture}[baseline=(current bounding box.center)]
	\matrix (m) [matrix of math nodes, row sep=\rowspacing, column sep=\colspacing, text height=0.8em,text depth=.2em] {
		 & & \rho \\
		\phantom{\psi} & \phantom{\psi} & \phantom{\psi} \\
	};
	\node[draw, rectangle, fit=(m-2-1)(m-2-3), inner sep=0] (W) {};
	\node at (W.center) {$W_{\pi}$};
	
	\draw (m-2-1.north) -- (m-1-1.north) -- ++(0,\overdist);
	\draw (m-2-1.south) -- ++(0,-\overdist);
	\draw (m-2-2.north) -- (m-1-2.north) -- ++(0,\overdist);
	\draw (m-2-2.south) -- ++(0,-\overdist);
	
	\draw (m-1-3) -- (m-2-3.north);
	\draw[rounded corners=\roundness] (m-2-3.south) -- ++(0,-\overdist) -- ++(\loopwidth,0) |- ($(m-1-3.north) + (0,\overdist)$) -- ++ (m-1-3);
\end{tikzpicture} \,\, \right] \\
&=& % ID
\frac{1}{(d+1)(d+2)}
\left[
\begin{tikzpicture}[baseline=(current bounding box.center)]
	\matrix (m) [matrix of math nodes, row sep=\rowspacing, column sep=\colspacing, text height=0.8em,text depth=.2em] {
		& & \rho \\
		\phantom{\psi} & \phantom{\psi} & \phantom{\psi} \\
	};
	
	\draw[rounded corners=\roundness] ($(m-1-1.north) + (0,\overdist)$) -- (m-2-1.north) -- (m-2-1.south) -- ++(0,-\overdist);
	\draw[rounded corners=\roundness] ($(m-1-2.north) + (0,\overdist)$) -- (m-2-2.north) --
	(m-2-2.south) -- ++(0,-\overdist);
	\draw[rounded corners=\roundness]  (m-1-3) -- (m-2-3.north) -- (m-2-3.south) -- ++(0,-\overdist) -- ++(\loopwidth,0) |-  ($(m-1-3.north) + (0,\overdist)$) -- (m-1-3);
\end{tikzpicture}
+ \!\!\!\!\!\!
% SWAP 13
\begin{tikzpicture}[baseline=(current bounding box.center)]
	\matrix (m) [matrix of math nodes, row sep=\rowspacing, column sep=\colspacing, text height=0.8em,text depth=.2em] {
		& & \rho \\
		\phantom{\psi} & \phantom{\psi} & \phantom{\psi} \\
	};
	
	\draw[rounded corners=\roundness] (m-1-3) -- (m-2-3.north) -- (m-2-1.south) -- ++(0,-\overdist);
	\draw[rounded corners=\roundness] ($(m-1-2.north) + (0,\overdist)$) -- (m-2-2.north) --
	(m-2-2.south) -- ++(0,-\overdist);
	\draw[rounded corners=\roundness] ($(m-1-1.north) + (0,\overdist)$) -- (m-2-1.north) -- (m-2-3.south) -- ++(0,-\overdist) -- ++(\loopwidth,0) |-  ($(m-1-3.north) + (0,\overdist)$) -- (m-1-3);
\end{tikzpicture}
+ \!\!\!\!\!\! % SWAP 23
\begin{tikzpicture}[baseline=(current bounding box.center)]
	\matrix (m) [matrix of math nodes, row sep=\rowspacing, column sep=\colspacing, text height=0.8em,text depth=.2em] {
		& & \rho \\
		\phantom{\psi} & \phantom{\psi} & \phantom{\psi} \\
	};
	
	\draw[rounded corners=\roundness] (m-1-3) -- (m-2-3.north) -- (m-2-2.south) -- ++(0,-\overdist);
	\draw[rounded corners=\roundness] ($(m-1-1.north) + (0,\overdist)$) -- (m-2-1.north) --
	(m-2-1.south) -- ++(0,-\overdist);
	\draw[rounded corners=\roundness] ($(m-1-2.north) + (0,\overdist)$) -- (m-2-2.north) -- (m-2-3.south) -- ++(0,-\overdist) -- ++(\loopwidth,0) |-  ($(m-1-3.north) + (0,\overdist)$) -- (m-1-3);
\end{tikzpicture}
+\!\!\!\!\!\!
% SWAP 12
\begin{tikzpicture}[baseline=(current bounding box.center)]
	\matrix (m) [matrix of math nodes, row sep=\rowspacing, column sep=\colspacing, text height=0.8em,text depth=.2em] {
		& & \rho \\
		\phantom{\psi} & \phantom{\psi} & \phantom{\psi} \\
	};
	
	\draw[rounded corners=\roundness] ($(m-1-1.north) + (0,\overdist)$) -- (m-2-1.north) -- (m-2-2.south) -- ++(0,-\overdist);
	\draw[rounded corners=\roundness] ($(m-1-2.north) + (0,\overdist)$) -- (m-2-2.north) --
	(m-2-1.south) -- ++(0,-\overdist);
	\draw[rounded corners=\roundness]  (m-1-3) -- (m-2-3.north) -- (m-2-3.south) -- ++(0,-\overdist) -- ++(\loopwidth,0) |-  ($(m-1-3.north) + (0,\overdist)$) -- (m-1-3);
\end{tikzpicture}
+\!\!\!\!\!\!
\begin{tikzpicture}[baseline=(current bounding box.center)]
	\matrix (m) [matrix of math nodes, row sep=\rowspacing, column sep=\colspacing, text height=0.8em,text depth=.2em] {
		& & \rho \\
		\phantom{\psi} & \phantom{\psi} & \phantom{\psi} \\
	};
	
	\draw[rounded corners=\roundness] (m-1-3) -- (m-2-3.north) -- (m-2-1.south) -- ++(0,-\overdist);
	\draw[rounded corners=\roundness] ($(m-1-1.north) + (0,\overdist)$) -- (m-2-1.north) --
	(m-2-2.south) -- ++(0,-\overdist);
	\draw[rounded corners=\roundness] ($(m-1-2.north) + (0,\overdist)$) -- (m-2-2.north) -- (m-2-3.south) -- ++(0,-\overdist) -- ++(\loopwidth,0) |-  ($(m-1-3.north) + (0,\overdist)$) -- (m-1-3);
\end{tikzpicture}
+\!\!\!\!\!\!
\begin{tikzpicture}[baseline=(current bounding box.center)]
	\matrix (m) [matrix of math nodes, row sep=\rowspacing, column sep=\colspacing, text height=0.8em,text depth=.2em] {
		& & \rho \\
		\phantom{\psi} & \phantom{\psi} & \phantom{\psi} \\
	};
	
	\draw[rounded corners=\roundness] (m-1-3) -- (m-2-3.north) -- (m-2-2.south) -- ++(0,-\overdist);
	\draw[rounded corners=\roundness] ($(m-1-2.north) + (0,\overdist)$) -- (m-2-2.north) --
	(m-2-1.south) -- ++(0,-\overdist);
	\draw[rounded corners=\roundness] ($(m-1-1.north) + (0,\overdist)$) -- (m-2-1.north) -- (m-2-3.south) -- ++(0,-\overdist) -- ++(\loopwidth,0) |-  ($(m-1-3.north) + (0,\overdist)$) -- (m-1-3);
\end{tikzpicture} \,\,\,\right] \\
&=& 
\frac{1}{(d+1)(d+2)}
\left[
\begin{tikzpicture}[baseline=(current bounding box.center)]
	\matrix (m) [matrix of math nodes, row sep=\rowspacing, column sep=0.8em, text height=0.8em,text depth=.2em] {
		& \\
		\phantom{\psi} & \phantom{\psi} \\
	};
	
	\draw[rounded corners=\roundness] ($(m-1-1.north) + (0,\overdist)$) -- (m-2-1.north) -- (m-2-1.south) -- ++(0,-\overdist);
	\draw[rounded corners=\roundness] ($(m-1-2.north) + (0,\overdist)$) -- (m-2-2.north) --
	(m-2-2.south) -- ++(0,-\overdist);
\end{tikzpicture}
+
\begin{tikzpicture}[baseline=(current bounding box.center)]
	\matrix (m) [matrix of math nodes, row sep=\rowspacing, column sep=0.8em, text height=0.8em,text depth=.2em] {
		\rho & \\
		\phantom{\psi} & \phantom{\psi} \\
	};
	
	\draw[rounded corners=\roundness] ($(m-1-1.north) + (0,\overdist)$) -- (m-1-1);
	\draw[rounded corners=\roundness] (m-1-1.south) -- (m-2-1.south) -- ++(0,-\overdist);
	\draw[rounded corners=\roundness] ($(m-1-2.north) + (0,\overdist)$) -- (m-2-2.north) --
	(m-2-2.south) -- ++(0,-\overdist);
\end{tikzpicture}
+
\begin{tikzpicture}[baseline=(current bounding box.center)]
	\matrix (m) [matrix of math nodes, row sep=\rowspacing, column sep=0.8em, text height=0.8em,text depth=.2em] {
		& \rho \\
		\phantom{\psi} & \phantom{\psi} \\
	};
	
	\draw[rounded corners=\roundness] ($(m-1-1.north) + (0,\overdist)$) -- (m-2-1.north) -- (m-2-1.south) -- ++(0,-\overdist);
	\draw[rounded corners=\roundness] ($(m-1-2.north) + (0,\overdist)$) -- (m-1-2);
	\draw[rounded corners=\roundness] (m-1-2) -- (m-2-2.south) -- ++(0,-\overdist);
\end{tikzpicture}
+
\begin{tikzpicture}[baseline=(current bounding box.center)]
	\matrix (m) [matrix of math nodes, row sep=\rowspacing, column sep=0.8em, text height=0.8em,text depth=.2em] {
		& \\
		\phantom{\psi} & \phantom{\psi} \\
	};
	
	\draw[rounded corners=\roundness] ($(m-1-1.north) + (0,\overdist)$) -- (m-2-1.north) -- (m-2-2.south) -- ++(0,-\overdist);
	\draw[rounded corners=\roundness] ($(m-1-2.north) + (0,\overdist)$) -- (m-2-2.north) --
	(m-2-1.south) -- ++(0,-\overdist);
\end{tikzpicture} +
\begin{tikzpicture}[baseline=(current bounding box.center)]
	\matrix (m) [matrix of math nodes, row sep=\rowspacing, column sep=0.8em, text height=0.8em,text depth=.2em] {
		\rho & \\
		\phantom{\psi} & \phantom{\psi} \\
	};
	
	\draw[rounded corners=\roundness] ($(m-1-1.north) + (0,\overdist)$) -- (m-1-1);
	\draw[rounded corners=\roundness] (m-1-1) -- (m-2-1.north) -- (m-2-2.south) -- ++(0,-\overdist);
	\draw[rounded corners=\roundness] ($(m-1-2.north) + (0,\overdist)$) -- (m-2-2.north) --
	(m-2-1.south) -- ++(0,-\overdist);
\end{tikzpicture} +
\begin{tikzpicture}[baseline=(current bounding box.center)]
	\matrix (m) [matrix of math nodes, row sep=\rowspacing, column sep=0.8em, text height=0.8em,text depth=.2em] {
		& \rho \\
		\phantom{\psi} & \phantom{\psi} \\
	};
	
	\draw[rounded corners=\roundness] ($(m-1-1.north) + (0,\overdist)$) -- (m-2-1.north) -- (m-2-2.south) -- ++(0,-\overdist);
	\draw[rounded corners=\roundness] ($(m-1-2.north) + (0,\overdist)$) -- (m-1-2);
	\draw[rounded corners=\roundness] (m-1-2) -- (m-2-2.north) --
	(m-2-1.south) -- ++(0,-\overdist);
\end{tikzpicture} \right] \\
&=& \frac{(I \otimes I + \rho \otimes I + I \otimes \rho)(W_{(1)(2)} + W_{(1\,2)})}{(d+1)(d+2)}
\end{eqnarray*}
        \caption{Second moment calculation in the special case $s = 1$.}
        \label{fig:second_moment}
    \end{figure}
    As before, we end up getting $I$ or $\rho$ for each position, depending on whether $1$ and $2$ were fixed by the permutation. The combinatorics is similar but not identical: there are $(s+1)!$ permutations which fix $1$, and of those, $s!$ fix $2$ and $s \cdot s!$ do not. Likewise, $s \cdot s!$ fix $2$ but not $1$. The remainder of the type A permutations fix neither $1$ nor $2$, and to count these we need a fact.
    \begin{fact}
    \label{fact:AB_bijection}
        Let $\pi \in \symm_n$ be a permutation, and consider $\pi' := (1 \, 2) \pi$. Then 
        $$
        \text{$1$ and $2$ are in distinct cycles of $\pi$} \iff \text{$1$ and $2$ are in the same cycle of $\pi'$}.
        $$
    \end{fact}
    In other words, there is a bijection between A and B permutations, so exactly half of all permutations ($\tfrac{1}{2} (s+2)!$) are type A, and half are type B. It follows that there are $\frac{s(s-1)}{2} \cdot s!$ type A permutations such that $\pi(1) \neq 1$ and $\pi(2) \neq 2$. Therefore, the overall contribution of the type A permutations is equal to $s!$ times
    $$
    (I \otimes I) + s(I \otimes \rho) + s(\rho \otimes I) + \frac{s(s-1)}{2} (\rho \otimes \rho) = (I + s \rho)^{\otimes 2} - \frac{s(s+1)}{2} \rho^{\otimes 2}.
    $$
    Fortunately, we do not have to repeat this counting argument for the type B permutations. The bijection from \Cref{fact:AB_bijection} decomposes each type B permutation $\pi$ as $(1\,2)\pi'$ where $\pi'$ is type A, and so 
    $$
    \Tr_{-1,2}(W_{\pi}(I \otimes \rho^{\otimes s})) = \Tr_{-1,2}(W_{(1\,2)}W_{\pi'}(I \otimes \rho^{\otimes s})) = W_{(1\,2)} \Tr_{-1,2}(W_{\pi'}(I \otimes \rho^{\otimes s})).
    $$
    Therefore, we can multiply our result for type A permutations by $W_{(1)(2)} + W_{(1\,2)} = 2 \sym^{(2)}$ to get the total. The result follows from some careful accounting of the scalar factors.
\end{proof}

\begin{cor}
    \label{cor:batch_estimate}
    Let $\rhohat = \frac{(d+s) \Psi - I}{s}$. For any observable $O \in \obs(B)$ and $\epsilon > 0$, we bound the probability of failure as 
    $$
    \Pr[|\Tr(O \rhohat) - \Tr(O \rho)| \geq \epsilon] \leq \frac{1}{\epsilon^2 s^2} \left[ \Tr( O^2 ) + 8s \Tr(O^2 \rho) \right].
    $$
\end{cor}
\begin{proof}
Our goal will be to compute the mean and variance of the estimate $\Tr(O \rhohat)$ in order to apply Chebyshev's inequality. By \Cref{lem:first_moment}, we have 
$$
\E[\rhohat] = \E \left[\frac{(d+s) \Psi - I}{s} \right] = \rho,
$$
and so the mean of the estimate is correct: $\E[\Tr(O \rhohat)] = \Tr(O \E[\rhohat]) = \Tr(O \rho)$.

To analyze the variance, let us first consider \emph{traceless} observables $O$, where $\Tr(O) = 0$. As usual, it will be useful to break the variance into a first moment term $\E[\Tr(O \Psi)]$ and a second moment term $\E[\Tr(O \Psi)^2]$:
\begin{align*}
\Var(\Tr(O \rhohat)) = \Var \left( \frac{(d+s)\Tr(O \Psi) - \Tr(O)}{s}\right)
= \left( \frac{d+s}{s} \right)^2 (\E[\Tr(O \Psi)^2] - \E[\Tr(O \Psi)]^2)
\end{align*}
Putting aside the $s^{-2}$ factor for now, the (squared) first moment term is
$$
(d+s)^2 \E[\Tr(O\Psi)]^2 = (d+s)^2 \left(\frac{\Tr(O) + s\Tr(O \rho)}{d+s}\right)^2 = s^2 \Tr(O \rho)^2.
$$
For the second moment term, we use $\E[\Tr(O \Psi)^2] = \E[\Tr((O\otimes O) (\Psi \otimes \Psi))]$ and \Cref{lem:second_moment} to write
\begin{align*}
(d+s)^2 \E[\Tr(O \Psi)^2] 
&= \frac{2(d+s)}{d+s+1} \Tr \left[ O^{\otimes 2}  \left( (I + s \rho)^{\otimes 2} - \frac{s(s+1)}{2} \rho^{\otimes 2} \right) \sym^{(2)} \right] \\
&\le \Tr \left[ O^{\otimes 2}  \left(I^{\otimes 2} + s (I \otimes \rho + \rho \otimes I) - \frac{s(s-1)}{2} \rho^{\otimes 2} \right) (2\sym^{(2)}) \right].
\end{align*}
Recall that $(2\sym^{(2)}) = W_{(1)(2)} + W_{(1\,2)}$, so to simplify, consider the contribution of those two terms:
\begin{align*}
    W_{(1)(2)}:&\;\;\Tr(O)^2 + 2s \Tr(O) \Tr(O \rho) + \frac{s(s-1)}{2} \Tr(O \rho)^2 \\
    W_{(1\,2)}:&\;\;\Tr(O^2) + 2s \Tr(O^2 \rho) + \frac{s(s-1)}{2} \Tr((O \rho)^2)
\end{align*}
In fact, because $\rho$ is pure, we have\footnote{If $\rho := \proj{\psi}$, we get
$\Tr((O\rho)^2) = \Tr(O\proj{\psi}O\proj{\psi}) = \Tr(\bra{\psi}O\ket{\psi}\bra{\psi}O\ket{\psi}) = \bra{\psi}O\ket{\psi}^2 = \Tr((O\rho)^2).$} that $\Tr((O\rho)^2) = \Tr(O \rho)^2$. Combining all of the above, we arrive at a bound for the (scaled) variance of $\Tr(O \Psi)$:
$$
(d+s)^2 \Var(\Tr(O \Psi)) 
\le \Tr(O^2) + 2s \Tr(O^2 \rho) - s \Tr(O \rho)^2 \\
\le \Tr(O^2) + 2s \norm{O^2}
$$
where the last inequality uses H\"older's inequality.\footnote{
H\"older's inequality: For density matrix $\rho$ and observable $O$,
$\Tr(O \rho) \le \sonenorm{O\rho} \leq (\inftynorm{O}\sonenorm{\rho}) \leq \inftynorm{O}.$
}

    It follows that $\Var(\Tr(O \rhohat))$ is 
    $$
    \Var(\Tr(O \rhohat)) \leq \frac{\Tr(O^2) + 2s \| O^2 \|}{s^2},
    $$
    \emph{for traceless observables}. Now suppose $O$ has nonzero trace, and let $O_0 := O - \Tr(O) I/d$ be its traceless part. Naturally, we have 
    \begin{align*}
    \Var(\Tr(O \rhohat)) 
    &= \Var(\Tr(O_0 \rhohat)) \leq \frac{\Tr(O_0^2) + 2s \| O_0^2 \|}{s^2}.
    \end{align*}
    Let's now show that we can bound each of those terms ($\Tr(O_0^2)$ and $\| O_0^2 \|$) using functions of the original observable $O$. First, for the $\Tr(O_0^2)$ term, we have that
    $$
    \Tr(O_0^2) = \Tr(O^2) - \Tr(O)^2/d \leq \Tr(O^2).
    $$
    Next consider the $\| O_0^2 \|$ term. We have that $\Tr(O) \le \|O\| d$, and so the largest eigenvalue (in absolute value) of $O - \Tr(O)I/d$ is at most $2\|O\|$. We get
    $$
    \| O_0^2 \| = \| O_0\|^2 \le (2\| O \|)^2 = 4 \| O^2 \|.
    $$
    Hence $\Var(\Tr(O \rhohat)) \leq \tfrac{1}{s^2} \left( \Tr(O)^2 + 8s \| O^2 \| \right)$, and the result follows by Chebyshev's theorem. 
\end{proof}

At last, we can prove the main theorem for this section.
\begin{proof}[Proof of Theorem~\ref{thm:ps_jm_ub}.]
    Consider an arbitrary observable $O$, and use Corollary~\ref{cor:batch_estimate} to bound the probability a single batch estimate $\rhohat^{(i)}$ is wrong by 
    $$
    \Pr[|\Tr(O \rhohat^{(i)}) - \Tr(O \rho)| \geq \epsilon] \leq \frac{1}{\epsilon^2 s^2} \left[ \Tr( O^2 ) + 8s \| O^2 \| \right] \leq \frac{B+8s}{\epsilon^2 s^2}. 
    $$
    Suppose we want this probability to be less than some constant $p < 1/2$; we leave it to the reader to check that at most 
    $\mathcal O(1/(\epsilon^2 p) + \sqrt{B/(\epsilon^2 p)})$
    samples suffice, and note that $s$ is chosen accordingly in Algorithm~\ref{alg:ps_jm_meas}.
    
    Recall that our final estimate $E$ is 
    $$
    E := \operatorname{median}(\Tr(O \rhohat^{(1)}), \ldots, \Tr(O \rhohat^{(k)})).
    $$
    Assume there are an odd number of batches, so the median is actually some $\Tr(O \rhohat^{(i)})$.
    If $E$ is a bad estimate, i.e., $|E - \Tr(O \rho)| > \epsilon$ then at least $k/2$ of the batch estimates are wrong: either $E$ and the estimates higher than it, or $E$ and the estimates lower than it.
    
    The batches are independent so Chernoff bounds the chance of seeing $\geq k/2$ failures.
    $$
    \Pr[|E - \Tr(O \rho)| \geq \epsilon] \leq \Pr \left[\# \{ i : |\Tr(O \rhohat^{(i)}) - \Tr(O \rho)| \geq \epsilon \} \geq \frac{k}{2} \right] \leq \sqrt{4p(1-p)}^{k}
    $$
    Setting this less than the failure probability $\delta$, we have 
    $$
    k \geq \frac{\log\delta^{-1}}{\log \, (4p(1-p))^{-1/2}}.
    $$
    Again, we note that $k$ is set accordingly in Algorithm~\ref{alg:ps_jm_meas}.
\end{proof}

\subsection{Discussion}

Let us compare this result with the original classical shadows protocol of Huang, Kueng, and Preskill \cite{Huang2020}. Their algorithm measures each copy of $\rho$ with $\mathcal{M}_1$\footnote{Technically, they use a Clifford unitary instead of a Haar random unitary, presumably for the sake of efficient implementation. However, $\mathcal{M}_1$ will work in place of their measurement, and the analysis is identical since it uses up to third moments of the ensemble of unitaries, which are the same for Clifford vs.\ Haar random unitaries, i.e., the Cliffords are a $3$-design~\cite{Webb2016, Zhu2017,kueng2015qubit,zhu2016clifford}.}, producing unbiased single-copy estimates $\rhohat_1, \ldots, \rhohat_s$ for $\rho$, which are then averaged into a batch estimate $\rhohat = \frac{1}{s} \sum_{i=1}^{s} \rhohat_i$. Given the observable $O$, the estimate is then $\Tr(O \rhohat)$, or the median of several batches, if necessary to reduce the probability of failure. 

We have just seen that the variance of a single-copy estimate is $\Var(\Tr(O \rhohat_i)) \leq \Tr(O^2) + \Tr(O^2 \rho)$, and averaging $s$ estimates together reduces the variance by a factor of $\frac{1}{s}$. On the other hand, our measurement with $\mathcal{M}_s$ provides an unbiased estimate with variance
\[
\Var(\Tr(O \rhohat)) \leq \frac{\Tr(O^2)}{s^2} + \frac{\Tr(O^2 \rho)}{s}.
\]
Since $\Tr(O^2 \rho)\leq 1$, we see that the quadratic denominator of $\Tr(O^2)$ (which is the dominant term) is making all the difference. 

% -------------------------------------------------------------
% SECTION - JOINT MEASUREMENT LOWER BOUND
% -------------------------------------------------------------

\section{Joint Measurement Lower Bound}
\label{sec:lower_bound}
\begin{thm}
$\task(B, \epsilon, \delta) = \Omega \left( \frac{\sqrt{B}}{\epsilon \log(B+1)} + \frac{\log \delta^{-1}}{\epsilon^2} \right)$ provided $B \le \epsilon d$.
\end{thm}
Notice that this bound matches the $\mathcal O((\sqrt{B}\epsilon^{-1} + \epsilon^{-2})\log(\delta^{-1}))$ upper bound up to a $\log(B)$ and a $\log(1/\delta)$ factor. We prove this as two separate lower bounds: $\Omega(\frac{\log \delta^{-1}}{\epsilon^2})$ and $\Omega( \frac{\sqrt{B}}{\epsilon \log(B+1)})$.

The first lower bound ($\Omega(\epsilon^{-2}\log(\delta^{-1})$) is derived via a reduction from the problem of distinguishing two pure states, $\rho_0$ and $\rho_1$, at trace distance $2 \epsilon$ from each other. We then use the known performance of the optimal measurement (Helstrom  measurement). 

The second lower bound is shown via a reduction from a problem in communication complexity known as Boolean Hidden Matching \cite{Gavinsky2007exponential}. We will show that any protocol for the classical shadows task implies a protocol for the Boolean Hidden Matching problem, which has known communication complexity lower bounds. These communication lower bounds will imply that the classical shadow must contain a significant amount of information. However, Holevo's theorem gives an upper bound on the amount of information gained through measurement. Therefore, in order to successfully complete the classical shadows task, many copies of the unknown state are required.

\subsection{\texorpdfstring{$\Omega(\epsilon^{-2}\log(\delta^{-1}))$}{Omega(epsilon\^{}-2)} lower bound}

The proof of the lower bound uses known results relating the trace distance between two states with our ability to distinguish the states by observables or binary measurements. In particular, the maximum gap for the expectation of a positive semi-definite observable is equal to the trace distance between the states:
\begin{lemma}
    \label{lem:observable_trace_distance}
    For arbitrary states $\rho$ and $\sigma$,
    $$
    \max_{0 \preceq O \preceq I} |\Tr( O \rho ) - \Tr( O \sigma )| = \tfrac{1}{2} \| \rho - \sigma \|_1.
    $$
    Furthermore, there is an optimal $O$ satisfying $\Tr(O^2) \leq \tfrac{1}{2} \rank (\rho - \sigma) \leq \tfrac{1}{2}(\rank \rho + \rank \sigma)$.
\end{lemma}
\begin{proof}
    Diagonalize $\rho - \sigma$ as $\sum_{i} \lambda_i \ketbra{\phi_i}{\phi_i}$. Define two positive semi-definite observables, $O_{+}$ and $O_{-}$.  
    \begin{align*}
        O_{+} &:= \sum_{i : \lambda_i > 0} \ketbra{\phi_i}{\phi_i} &
        O_{-} &:= -\sum_{i : \lambda_i < 0} \ketbra{\phi_i}{\phi_i}
    \end{align*}
    Clearly $O_{+} - O_{-}$ is a projector onto the eigenvectors of $\rho - \sigma$, so $\Tr((O_{+} - O_{-})(\rho - \sigma)) = \Tr(\rho - \sigma) = 0$, and $\rank (O_{+} - O_{-}) = \rank (\rho - \sigma)$. On the other hand, 
    $$
    \Tr((O_{+} + O_{-})(\rho - \sigma)) = \Tr \left( \sum_{i : \lambda_i \neq 0} (\sgn \lambda_i) \lambda_i \ketbra{\phi_i}{\phi_i} \right) = \sum_{i : \lambda_i \neq 0} |\lambda_i| = \| \rho - \sigma \|_{1}.
    $$ 
    It follows that $\Tr(O_{+} (\rho - \sigma)) = \Tr(O_{-}(\rho - \sigma)) = \tfrac{1}{2} \| \rho - \sigma \|_1$. Since $O_{+}$ and $O_{-}$ are orthogonal, the rank of their sum, $\rank (\rho - \sigma)$, is the sum of their ranks. Hence, we can take whichever of $O_{+}$ and $O_{-}$ has rank at most $\tfrac{1}{2} \rank (\rho - \sigma)$. 
    \end{proof}
    
Separately, we know the optimal measurement for distinguishing a uniformly randomly chosen $\rho$ or $\sigma$ is given by:
\begin{lemma}[Helstrom measurement \cite{helstrom1976quantum}]
    \label{lem:helstrom}
    The optimal measurement for distinguishing states $\rho$ and $\sigma$ succeeds with probability
    $\frac{1}{2} + \frac{1}{4} \| \rho - \sigma \|_1$.
\end{lemma}

\begin{theorem}
\label{thm:distinguishing_lb}
$\task(1, \epsilon, \delta) = \Omega(\epsilon^{-2} \log(1/\delta))$.
\end{theorem}
\begin{proof}
Let $\rho_0$ and $\rho_1$ be pure states with trace distance $\tfrac{1}{2} \| \rho_0 - \rho_1 \|_1 = 2 \epsilon$.
We claim that a protocol for the classical shadows task to $\epsilon$-approximate the expected values of observables of rank 1 with probability of failure at most $\delta$ can be used to distinguish states $\rho_0, \rho_1$ with probability of failure at most $\delta$. First, apply the measurement subroutine to unknown state $\rho_b^{\otimes s}$ to produce the classical shadow. Then, use the observable $O$ from Lemma~\ref{lem:observable_trace_distance} to estimate $\Tr(O \rho_b)$. If the estimate is closer to $\Tr(O \rho_0)$ then guess $b=0$, otherwise guess $b=1$. Since the gap $|\Tr(O \rho_0) - \Tr(O \rho_1)| = \tfrac{1}{2} \| \rho_0 - \rho_1 \| = 2\epsilon$, we succeed whenever the gap between the estimate and $\Tr(O \rho_b)$ is less than $\epsilon$. 

We can see this classical shadows protocol as a binary measurement distinguishing $\rho_0$ and $\rho_1$. Since it succeeds with probability $1 - \delta$, the optimal distinguishing measurement from Lemma~\ref{lem:helstrom} must do better, so 
    \begin{align*}
        \left(1 - 2\delta \right)^2 &\leq \frac{1}{4} \| \rho_0^{\otimes s} - \rho_1^{\otimes s} \|_1^2 
        = 1 - \Tr( \rho_0^{\otimes s} \rho_1^{\otimes s} ) 
        = 1 - \Tr( \rho_0 \rho_1 )^{s} \\
        &= 1 - \left(1 - \frac{1}{4} \| \rho_0 - \rho_1 \|_1^{2} \right)^{s} 
        \!= 1 - (1 - \epsilon^2)^{s}
    \end{align*}
where we have used the equation $\frac{1}{2} \| \rho_0 - \rho_1 \| = \sqrt{1 - \Tr( \rho_0 \rho_1 )}$ relating trace distance and fidelity for pure states~\cite{Fuchs1999}. Rearranging, we have $(1 - \epsilon^2)^{s} \leq 1 - (1 - 2 \delta)^2 = 4 \delta - 4 \delta^2 \leq 4 \delta$. Taking logs and using $1 - \frac{1}{x} \leq \ln x$ we have 
$$
-\frac{s\epsilon^2}{1 - \epsilon^2} \leq s \log(1 - \epsilon^2) \leq \log (4\delta),
$$
or equivalently,
$$
s \geq \frac{1-\epsilon^2}{\epsilon^2} \log \left(\frac{1}{4 \delta}\right) = \Omega(\epsilon^{-2} \log(1/\delta)). 
$$
\end{proof}

\subsection{\texorpdfstring{$\Omega(\epsilon^{-1} B / \log(B+1))$}{Omega(epsilon\^{}-1 B/log(B+1))} lower bound}

To prove this lower bound, we leverage the perspective that the classical shadows task is fundamentally a one-way communication problem---recall the setup of the classical shadows task (c.f., \Cref{sec:classical_shadows_task}) where Melanie measures copies of an unknown state $\rho$ and sends a classical message to Esteban that allows him to estimate the expectation of some observable $O$ on $\rho$. Intuitively, measuring more copies of $\rho$ means the message will contain more information about $\rho$. Conversely, if we can prove that Melanie's message must contain a lot of information, we can prove that she must have measured many copies of $\rho$. In other words, there is a tight correspondence between the sample complexity and one-way classical communication complexity of the classical shadows task. 

Formalizing this correspondence is somewhat tricky, so we leave the precise details for later (in particular, \Cref{sec:one-way_communication_complexity}). However, once this correspondence is established, the high-level structure of the proof is relatively straightforward.

Our starting point is a one-way communication task called ``Boolean Hidden Matching''. As in the classical shadows task, there are two parties involved in the task: Alice and Bob. Alice has a labeled graph and Bob has a ``partial matching'' (a collection of vertex-disjoint edges from the graph). Together, these encode a secret bit\footnote{See \Cref{sec:boolean_hidden_matching} for the precise definition.}. Bob doesn't know the labels on the graph, so Alice's goal is to send him a classical message so that he can extract the encoded bit. \cite{Gavinsky2007exponential} shows a lower bound on the number of bits that Alice must send to be successful---namely, she must send $\Omega(\sqrt{n/\alpha})$ bits where $n$ is the number of vertices in Alice's graph and $\alpha$ is the fraction of edges in Bob's partial matching. 

Our goal will be to take this lower bound for Boolean Hidden Matching and turn it into a lower bound for the classical shadows task. To do this, we create an ensemble of states (corresponding to labeled graphs) and observables (corresponding to partial matchings) such that computing the expected value of an observable with a state solves the Boolean Hidden Matching problem for the corresponding graph and matching. In other words, if Alice and Bob want to solve the Boolean Hidden Matching problem, they can first create the corresponding states and observables, and then use a protocol for classical shadows. We give a depiction of this reduction in \Cref{fig:lower_bound_figure}.

\begin{figure}
    \centering
    \includegraphics[scale=.6]{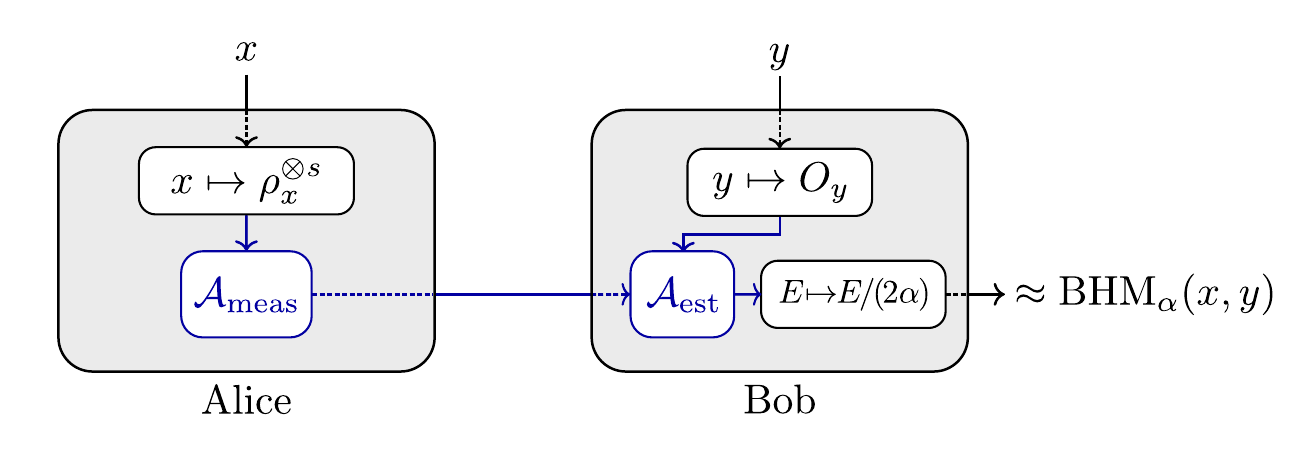}
    \caption{Protocol for Boolean Hidden Matching using classical shadows subroutines (shown in blue). From her input $x$, Alice prepares $\rho_x^{\otimes s}$, measures, and sends the classical shadow to Bob. From his input $y$, Bob computes $O_y$ and then estimates $\Tr(O_y \rho_x)$ to accuracy $\alpha$ from the classical shadow that Alice sent to him. He then uses that estimate to answer the Boolean Hidden Matching problem. Correctness follows from the fact that $\Tr(O_y \rho_x) = 2\alpha \cdot \BHM_\alpha(x,y)$. For the details of $\rho_x$ and $O_y$ see \Cref{thm:bhm_to_shadows}.}
    \label{fig:lower_bound_figure}
\end{figure}

To tie everything together, we appeal to the equivalence between the sample complexity of the classical shadows task and the one-way communication complexity. Namely, Alice must measure a number of copies of her state (roughly) proportional to the number of bits she wants to send Bob. Since we have a lower bound on the number of bits she must send Bob, we have a lower bound on the number of copies she must measure. This completes the proof. 

This overall idea draws considerable inspiration from that in the work of Gosset and Smolin for compressing classical descriptions of quantum states \cite{gosset_smolin2019}. Our proof can be seen as generalization of their techniques.

The remainder of this section is devoted to formalizing the above ideas. \Cref{sec:one-way_communication_complexity} introduces one-way communication complexity, culminating in a powerful theorem connecting the number of bits exchanged in a communication protocol and the amount of Shannon information exchanged in the protocol. In \Cref{sec:boolean_hidden_matching}, we give the formal definition of the Boolean Hidden Matching problem and fill in the missing details from the proof outline above.

\subsubsection{One-way communication complexity}
\label{sec:one-way_communication_complexity}
Because our proof is based on principles from communication complexity, let's briefly introduce that topic. We are interested in \emph{one-way communication protocols} where two parties---Alice and Bob---are trying to jointly compute some function $f \colon \mathcal X \times \mathcal Y \to \{0,1\}$. Alice is given some input $x \in \mathcal X$ and Bob is given some input $y \in \mathcal Y$. Alice's goal is to send a single message $m \in \{0,1\}^*$ to Bob, so that he can compute $f(x,y)$. Of course, she could choose to send her entire input $x$, but in many cases it may be possible to communicate fewer bits and still be successful.

To be precise about the size of the message Alice must send, let $X$ and $Y$ be the random variables (possibly correlated) for the inputs of Alice and Bob, respectively, and let $M$ be the random variable for Alice's message to Bob. Notice that implicit in $M$ is Alice's communication strategy, which may be an arbitrary (randomized) function of her input. Let's start with the easiest setting, where Alice and Bob run deterministic algorithms. 
\begin{defn}[Deterministic One-Way Communication Complexity]
$D^{(X,Y)}_\delta(f)$, the bounded-error deterministic one-way communication complexity of $f$, is the minimum number of bits that Alice must send to Bob to compute $f$ with at most $\delta$ probability of error whenever their inputs are chosen according to the distribution $(X,Y)$.
\end{defn}
A natural variant of classical one-way protocols is when Alice and Bob are allowed to run randomized algorithms. There are two settings: private-coin protocols, where Alice and Bob each have access to private random strings; and public-coins protocols, where Alice and Bob also have access to a shared random string along with their private strings. It will not be critical to completely understand the nuances of the various types of protocols for our proof, but we define them in order to precisely state the theorems on which the lower bound rests.
\begin{defn}[Randomized One-Way Communication Complexity]
$R_\delta(f)$, the bounded-error randomized one-way communication complexity of $f$, is the minimum number of bits that Alice must send to Bob with a public-coin protocol to compute $f$ over \emph{all} possible inputs with failure probability at most $\delta$.
\end{defn}

While randomized strategies may seem more powerful than deterministic strategies, Yao's minimax principle shows that there is always some input distribution for which the randomized and deterministic complexities coincide:
\begin{thm}[Yao's minimax principle \cite{yao1977probabilistic}]
\label{thm:yaos_minimax}
$\max_{(X,Y)} D^{(X,Y)}_\delta (f) = R_\delta(f)$. 
\end{thm}

It turns out that we will eventually be interested in the amount of \emph{information} contained in Alice's message $M$, not just the length, which is what is measured by the communication complexity. That said, these two quantities are intuitively related---if the information $I(M : X)$ that Alice's message $M$ reveals about her input $X$ is much lower than the number of bits she is communicating, she should be able to send a smaller message and still be successful. The following theorem formalizes this message compression idea:

\begin{thm}[\cite{harsha2007communication}] % Harsha et al.
\label{thm:information_bounds_communication}
$D^{(X,Y)}_{\Delta+\delta}(f) = 2 \Delta^{-1} \left[ \min I(M : X) + O(1) \right]$ where the minimization is over all one-way private-coin protocols for $f$ with input distribution $(X,Y)$ and probability of error at most $\delta$. 
\end{thm}

In the next section, we will show that classical shadows must also contain a lot of information, which will be the basis of our lower bound.

\subsubsection{Classical shadows for Boolean Hidden Matching}
\label{sec:boolean_hidden_matching}
Our starting point is a lower bound for the one-way communication complexity for the Boolean Hidden Matching function $\BHM_\alpha \colon \mathcal X \times \mathcal Y \to \{0,1\}$ which was introduced by Gavinsky, Kempe, Kerenidis, Raz, and de Wolf. 

\begin{thm}[\cite{Gavinsky2007exponential}]
\label{thm:gavinsky_main}
$R_\delta(\BHM_\alpha) = \Theta(\sqrt{n/\alpha})$ for any constant $\delta < 1/2$.
\end{thm}

Recall that our lower bound technique is to show that a classical shadows strategy with few samples implies a communication protocol for the Boolean Hidden Matching function with low complexity. To do this, let's look more closely at the $\BHM_\alpha$ function.

It will be useful describe the function directly as a communication problem with the inputs $x \in \mathcal X$ for Alice and $y \in \mathcal Y$ for Bob. Alice is given $(0,1)$-assignments for the $n$ vertices of a graph, and Bob is given $(0,1)$-assignments to $\alpha n$ vertex-disjoint edges in the graph. A set of vertex-disjoint edges from a graph is called a \emph{matching}, hence the name of the function. Importantly, the function is only defined on inputs for Alice and Bob that satisfy the following promise: for each edge in the matching, the parity of the connected vertices (from Alice's input) plus Bob's edge bit assignment is some constant $b\in \set{0,1}$. The output of the function is then defined as this bit $b$.

Formally, for every $n \ge 2$, the Boolean Hidden Matching function $\BHM_\alpha \colon \mathcal X \times \mathcal Y \to \{0,1\}$ with parameter $\alpha \in (0,1/4]$ is defined on inputs $\mathcal X = \{ 0, 1 \}^{n}$ and $\mathcal Y = (\mathbb N, \mathbb N)^{\alpha n} \times \{ 0, 1 \}^{\alpha n}$ as follows:

\begin{center}
    \begin{tabular}{p{1.5cm} p{11.5cm}}
    \textbf{Alice:} & $x \in \mathcal X$. \\
    \textbf{Bob:} & $y = (\mathcal M, w) \in \mathcal Y$, where $\mathcal{M} = \{ (i_1, j_1), \ldots, (i_{\alpha n}, j_{\alpha n})\}$ is a matching and $w \in \{ 0, 1 \}^{\alpha n}$ is a bit assignment. \\
    \textbf{Promise:} & There exists $b \in \{ 0, 1 \}$ such that $b = x_{i_k} \oplus x_{j_k} \oplus w_{k}$ for all $k$. \\
    \textbf{Output:} & $b \in \{0,1\}$
    \end{tabular}
\end{center}

In the setting where $\alpha$ is constant, \cite{Gavinsky2007exponential} show that the \emph{quantum} communication complexity of the Boolean Hidden Matching problem is low---Alice only needs to send a $(\log n)$-qubit state. Gosset and Smolin \cite{gosset_smolin2019} notice that this implies the existence of a set of states and observables whose expectation values give solutions to the Boolean Hidden Matching function. We generalize this observation to the non-constant $\alpha$ setting below:

\begin{thm}
\label{thm:bhm_to_shadows}
There is a set of states $\{\rho_x \in \mathbb C^n\}_{x \in \mathcal X}$ and observables $\{O_y \in \mathrm{Obs}(\alpha n)\}_{y \in \mathcal Y}$ such that $\Tr(O_y \rho_x) = 2 \alpha \cdot \BHM_\alpha(x,y)$. Furthermore, a protocol for the classical shadows task for observables of squared Frobenius norm $B:= \alpha n$, estimation accuracy $\epsilon := \alpha$, and failure probability $\delta$ implies a one-way private-coin protocol for Boolean Hidden Matching with failure probability $\delta$.
\end{thm}
\begin{proof}
Given valid inputs $x$ and $y = (\mathcal M, w)$ to the $\BHM_\alpha$ function, define the pure state 
$$
\ket{\psi_x} := \frac{1}{\sqrt{n}} \sum_{i=1}^n (-1)^{x_i} \ket{i}
$$
and the observable
$$
O_y := \sum_{k=1}^{\alpha n} \frac{1}{2}(\ket{i_k} - (-1)^{w_k} \ket{j_k})(\bra{i_k} - (-1)^{w_k} \bra{j_k}).
$$
Notice that $O_y \in \obs(\alpha n)$ since $O_y$ is a $\alpha n$-rank projector. Letting $\rho_x := \ketbra{\psi_x}{\psi_x}$, we get
$$
\Tr( O_y \rho_x ) = \bra{\psi_x} O_y \ket{\psi_x} = \frac{1}{n} \sum_{k=1}^{\alpha n}(1 - (-1)^{x_{i_k} \oplus x_{j_k} \oplus w_k })= 2 \alpha b = 2\alpha \BHM_\alpha(x,y).
$$
In particular, this implies that if $E$ is an $\alpha$-approximation to $\Tr( O_y \rho_x )$, then 
$$
|E - \Tr( O_y \rho_x )| < \alpha \implies \left|\frac{E}{2\alpha} - \BHM_\alpha(x,y)\right| < 1/2,
$$
or in other words, rounding $E/(2\alpha)$ is equal to $\BHM_\alpha(x,y)$.

We now claim that the existence of these states and observables implies a private-coin one-way protocol for the Boolean Hidden Matching problem (see \Cref{fig:lower_bound_figure}): Suppose we want a protocol for $\BHM_\alpha$ with probability of failure at most $\delta$. Let $s = \task(\alpha n, \alpha, \delta)$. On input $x$, Alice prepares the state $\rho_x^{\otimes s}$, measures it with a valid classical shadows strategy, and sends the resulting classical shadow to Bob. On input $y$, Bob computes the observable $O_y$, and then computes an estimate $E$ for $\Tr(O_y \rho_x)$ using the classical shadow sent by Alice. The correctness of the classical shadows strategy implies that $E$ is an $\alpha$-approximation to $\Tr(O_y \rho_x)$ with probability of failure at most $\delta$. As shown above, Bob can then compute $\BHM_\alpha(x,y)$ with failure probability at most $\delta$ by appropriately rounding the estimate.
\end{proof}

Let us now note a key property of the one-way protocol in \Cref{thm:bhm_to_shadows} for the Boolean Hidden Matching problem. Namely, once Alice prepares $\rho_x^{\otimes s}$, she no longer uses her original input $x$. Her message (the classical shadow) only depends on her measurement of the state $\rho_x^{\otimes s}$. In particular, if her message is to contain a lot of information about her input $x$, then Holevo's theorem stipulates that she must be measuring a state of high dimension, or, in other words, $s$ must be large:

\begin{thm}[Holevo \cite{holevo1973bounds}]
\label{thm:holevo}
Let $Z$ be the classical outcome of measuring a $d$-dimensional state drawn from an ensemble $\{ \rho_x \}_{x \in \mathcal X}$ according to $x \sim X$. Then, $I(X : Z) \le \log d$.
\end{thm}

Na\"{i}vely, the states $\rho_{x}^{\otimes s}$ in \Cref{thm:bhm_to_shadows} consist of $s$ qudits of dimension $n$, i.e., it is a space of dimension $n^s$. However, since each $\rho_x^{\otimes s}$ is invariant under permutation, it belongs to the \emph{symmetric subspace}, which has dimension nearly a factor of $s!$ smaller (see \Cref{symmetric_subspace_size}).

We are now ready to put all of the pieces of the lower bound together:
\begin{thm}
\label{thm:main_lower_bound}
$\task(B, \epsilon) = \Omega\left(\frac{\sqrt{B}}{\epsilon \log(B+1)}\right)$ provided $B \le \epsilon d$.
\end{thm}
\begin{proof} Using the communication complexity of $\BHM_\alpha$ as our starting point, we first show that Alice's message to Bob in every successful protocol for the Boolean Hidden Matching problem must contain a significant amount of information. To show this, note that by Yao's minimax principle (Theorem~\ref{thm:yaos_minimax}), there exists a distribution $(X,Y)$ such that $R_\delta(\BHM_{\alpha}) = D^{(X,Y)}_\delta(\BHM_{\alpha})$, for any given $\delta$ is the probability of error. \Cref{thm:information_bounds_communication} lets us upper bound the (deterministic) complexity with mutual information, and \Cref{thm:gavinsky_main} proves a lower bound. Thus,  
\begin{align*}
D^{(X,Y)}_{\delta}(\BHM_\alpha) &= O(\min I(M : X) + 1), \text{ and} \\
D^{(X,Y)}_{\delta}(\BHM_\alpha) &= \Omega(\sqrt{n/\alpha}),
\end{align*}
for any constant $\delta$. It follows that $I(M : X) = \Omega(\sqrt{n/\alpha})$ for any one-way private-coin protocol for $\BHM_\alpha$.

Now consider the classical shadows strategy for solving $\BHM_\alpha$ as described by \Cref{thm:bhm_to_shadows}, and suppose that Alice measures $s = \task(\alpha n, \alpha)$ copies of $\rho_x$.\footnote{It is possible that Alice could use fewer than $\task(\alpha n, \alpha)$ samples for this specific application to Boolean Hidden Matching, but it will be useful later that $s$ is big enough to estimate a broader class of observables (specifically those used in \Cref{thm:distinguishing_lb}). }  Recall that Alice's message $M$ depends only on her measurement of $\rho_X^{\otimes s}$ which has classical outcome $Z$. We get
$$
I(X : M) \leq I(X : Z) \leq \log(\dim \sym^{(s)}) \le O(s \log(n/s + 1)),
$$
where we have used (in order) the Data Processing Inequality, Holevo's theorem (\Cref{thm:holevo}), the dimension of the symmetric subspace (\Cref{symmetric_subspace_size}), and the following inequality:
$$
\dim \sym^{(s)} = \binom{n+s-1}{n-1} \leq \binom{n+s}{n} = \binom{n+s}{s} \leq \left( \frac{e(n+s)}{s} \right)^{s}.
$$

Notice that we now have both an upper bound and a lower bound for the mutual information between Alice's input and her message for a one-way protocol for $\BHM_\alpha$ with constant error probability. Setting $\epsilon := \alpha$ and $B := \epsilon n$, we have 
\begin{align*}
    I(X : M) &= \Omega(\sqrt{B}/\epsilon), \text{ and} \\
    I(X : M) &= O(s \log(B/(\epsilon s) + 1)).
\end{align*}
It follows that 
\[
s = \Omega\left(\frac{\sqrt{B}}{\epsilon \log(B/(\epsilon s) + 1)}\right).
\]
Notice that if we substitute any lower bound for $s$ in the RHS of the equation above, then we get a new lower bound for $s$ on the LHS. Unfortunately, plugging in the trivial lower bound ($s \ge 1$) is not very tight. Instead, we will use the $s = \Omega(1/\epsilon^2)$ lower bound from the previous section. To justify this, notice that 
\[
s = \task(B,\epsilon) \ge \task(1,\epsilon) = \Omega(1/\epsilon^2),
\]
where we have used \Cref{thm:distinguishing_lb} and the fact that $B \ge 1$ since $\| O \| = 1$. Therefore, we can plug $s = \Omega(1/\epsilon^2)$ into the RHS above to arrive at the following:
$$
s = \Omega\left(\frac{\sqrt{B}}{\epsilon \log(B\epsilon + 1)}\right).
$$
Assuming $\epsilon \leq 1$, we simplify this to $s = \Omega( \frac{\sqrt{B}}{\epsilon \log(B+1)})$.\footnote{While this simplification apparently makes lower bound weaker for no reason, it doesn't actually effect the overall lower bound. In regimes where $\log(B\epsilon + 1)$ is significantly smaller than $\log(B + 1)$, the additive $1/\epsilon^2$ term in the lower bound becomes dominant.}

Finally, we point out that the construction of \Cref{thm:bhm_to_shadows} operates in the regime where the states have dimension $d := n$ and the observables are of rank $B = \epsilon d$. One can extend the lower bound to apply to all observables of rank $B \le \epsilon d$ by embedding the states used in \Cref{thm:bhm_to_shadows} into a subspace of dimension $n \leq d$ and keeping the observables the same.
\end{proof}

% -------------------------------------------------------------
% SECTION - INDEPENDENT MEASUREMENT UPPER BOUND
% -------------------------------------------------------------

\section{Independent Measurement Upper Bound}
\label{sec:independent_measurements}
Since the global Clifford group acting on qubits is a $3$-design, the randomized Clifford measurement classical shadows algorithm of Huang, Kueng, and Preskill \cite{Huang2020} can be viewed as simulating independent $\mathcal{M}_1$ measurements on all copies of $\rho$ then, constructing an unbiased estimator from the measurement outcome on each copy. Their result is for independent measurements and general mixed states, but it upper bounds pure states as a special case. 
\begin{thm}[Huang, Kueng, Preskill \cite{Huang2020}] 
For all $\epsilon, \delta > 0$, 
    $$\itask(B, \epsilon, \delta) = \bigo{ \frac{B\log(\delta^{-1})}{\epsilon^2}}.$$
\end{thm}
Huang, Kueng, and Preskill also show a matching lower bound, but the hard instances they construct are with states of full rank. We give an independent measurement classical shadows algorithm for pure states which is better in certain parameter regimes (and is no worse). 
\begin{thm}
\label{thm:ps_im_ub}
    For all $\epsilon, \delta > 0$,
    $$\itask(B, \epsilon, \delta) = \bigo{\min \left\{ \frac{B}{\epsilon^2}, \frac{\sqrt{Bd}}{\epsilon} + \frac{1}{\epsilon^2} \right \} \log(\delta^{-1})}.$$
\end{thm}

For example, consider the parameter regime in which $\delta$ is a constant, $B = d$, and any $\epsilon = o(1)$. Note that this encompasses natural settings such as estimating full-weight Paulis. One can check that (as $d$ grows) the sample complexity given by \Cref{thm:ps_im_ub} is $\mathcal O(d/\epsilon + 1/\epsilon^2)$, which is evidently less than $\mathcal O(d/\epsilon^2)$, the sample complexity of the Huang-Kueng-Preskill protocol. In general, our approach gives lower sample complexity whenever $\epsilon = o(\sqrt{d/B})$ and $B = \omega(1)$.

As it turns out, our measurement algorithm is the same as the one in \cite{Huang2020}---on each copy of $\rho$, we make an independent measurement with the POVM $\mathcal{M}_1$, which (on multi-qubit systems) can be performed with a random Clifford measurement since we only use third moments of the Haar measure. The difference is in how we construct our estimator for the unknown state. To see this, first let $\Psi_1, \ldots, \Psi_s$ be the Hermitian random variables for the measurement outcomes. Using Lemma~\ref{lem:first_moment}, notice that $\rhohat_{i} := (d+1)\Psi_i - I$ is an unbiased estimator for the unknown state, i.e., $\E[\rhohat_i] = \rho.$ The average of the $\rhohat_i$'s, i.e., 
$$
\hat{X} := \frac{1}{s} \sum_{i=1}^{s} \rhohat_i
$$
is effectively the Huang-Kueng-Preskill estimator. Our key observation is that when $\rho$ is pure, $\rhohat_i \rhohat_j$ is also an unbiased estimator of $\rho$:
$$
\E[\rhohat_i \rhohat_j] = \E[\rhohat_i] \E[\rhohat_j] = \rho^2 = \rho.
$$
where we have used the independence of the measurements for the first equality and the purity of $\rho$ for the last.
In light of this, we consider an estimator $\hat{Y}$ defined to be the average of the $s(s-1)$ quadratic terms where $i \neq j$.
$$
\hat{Y} := \frac{1}{s(s-1)} \sum_{i \neq j} \rhohat_i \rhohat_j.
$$
To analyze the accuracy of the estimator $\hat{Y}$, we will once again turn to Chebyshev's inequality: 
\begin{align*}
\Pr[|\Tr(O \hat{Y}) - \Tr(O \rho)]| \geq \epsilon] \leq \frac{\Var(\Tr(O \hat{Y}))}{\epsilon^2}.
\end{align*}
Expanding out the variance term using the definition of $\hat Y$, we get
\[
\Var( \Tr(O \hat{Y}) ) = \frac{1}{s^2 (s-1)^2} \sum_{i \neq j} \sum_{k \neq \ell} \Cov( \Tr(O \rhohat_i \rhohat_j), \Tr(O \rhohat_k \rhohat_{\ell})).
\]

We need to bound all of these covariance terms to bound the variance. When all indices $i,j,k,\ell$ are distinct, then the covariance is 0 (by independence). For the other four cases, we rely on corollaries~\ref{cor:ijjk}, \ref{cor:ijkj}, \ref{cor:ijji}, and \ref{cor:ijij}, which we summarize in the following lemma (proof in \Cref{subsec:covariance_bounds}):
\begin{lemma} 
\label{lem:all_the_covariances}
For each combination of $i,j,k,\ell$, $\Cov(\Tr(O \rhohat_i \rhohat_j), \Tr(O \rhohat_k \rhohat_{\ell}))$ is
\begin{enumerate}
    \item One index matches $(|\{ i, j \} \cap \{ k, \ell \}| = 1)$
    \begin{itemize}
        \item Match in different positions $(i = \ell$ or $j = k)$: $\mathcal O(\|O\|^2)$ 
        \item Match in same position $(i = k$ or $j = \ell)$: $\mathcal O(\|O\|^2)$
    \end{itemize}
    \item Both indices match $(|\{ i, j \} \cap \{ k, \ell \}| = 2)$
    \begin{itemize}
        \item Order swapped $(i = \ell$ and $j = k)$: $\mathcal O(Bd)$
        \item Same order $(i = j$ and $k = \ell)$: $\mathcal O(Bd)$ 
    \end{itemize}
\end{enumerate}
\end{lemma}

Since $\|O\|^2 \le 1$, the contribution from the first two terms is extremely small compared to the last term. This gives us the following bound on the variance:

\begin{lemma}
\label{lem:quadratic_covariance}
$\Var(\Tr(O \hat{Y})) = \mathcal O(\frac{Bd}{s^2} + \frac{1}{s})$.
\end{lemma}
\begin{proof}
Expand the variance as
\[
\Var( \Tr(O \hat{Y}) ) = \frac{1}{s^2 (s-1)^2} \sum_{i \neq j} \sum_{k \neq \ell} \Cov( \Tr(O \rhohat_i \rhohat_j), \Tr(O \rhohat_k \rhohat_{\ell})).
\]
Using \Cref{lem:all_the_covariances}, we account for the contribution of each type of covariance term to get
$$
\Var(\Tr(O \hat{Y})) = \bigo{\frac{0\cdot s^4 + \|O\|^2 \cdot s^3 + \|O\|^2 \cdot s^3 + Bd \cdot s^2 + Bd \cdot s^2}{s^4}} = \bigo{Bd/s^2 + 1/s}
$$
where we have used that $\|O\|^2 \le 1$ and that there are $\mathcal O(s^4)$ terms where $i,j,k,\ell$ are distinct; $\mathcal O(s^3)$ terms where exactly one index matches; and $\mathcal O(s^2)$ terms where both indices match. 
\end{proof}

Putting everything together, we can now prove the claimed sample complexity in Theorem~\ref{thm:ps_im_ub}. 

\begin{proof}[Proof of Theorem~\ref{thm:ps_im_ub}]
We first point out that the sample complexity of our new estimator is only better (or at least no worse) when $\epsilon \leq \sqrt{B/d}$, so when that does not hold we simply use the original $\hat{X}$ estimator of Huang, Kueng, and Preskill.\footnote{Actually, it is better to smoothly transition between the estimators $\hat{X}$ and $\hat{Y}$ using a convex combination rather than use a sharp threshold. However, the improvement is only a constant factor and would require computing covariance of linear vs.\ quadratic terms (e.g., $\Cov(\Tr(O \rhohat_i \rhohat_j), \Tr(O \rhohat_k))$) to justify rigorously. Hence, we simply use one or the other.}

Otherwise, we use Algorithm~\ref{alg:ps_im_meas} to measure the state and construct several $\hat{Y}$ estimators (line~\ref{state:estimator}, constituting the classical shadow. This shadow is then used in Algorithm~\ref{alg:ps_jm_est} for the observable estimation step, which once again uses the median-of-means method where the analysis will be identical to that in the proof of \Cref{thm:ps_jm_ub}.

It suffices to analyze the variance of the estimator constructed within each batch, which is $\bigo{Bd/s^2 + 1/s}$ by \Cref{lem:quadratic_covariance}. To apply Chebyshev's inequality, we need the variance to be at most $\epsilon^2$, which occurs when we have at least $s = \mathcal O(\sqrt{Bd}/\epsilon + 1/\epsilon^2)$ samples.
\end{proof}

We simulate the new quadratic estimator as shown in \Cref{fig:linear-vs-quadratic-simulation}. The plots show the empirical variances of the linear and quadratic estimators in a regime where the target observable has large Frobenius norm, namely $B = d$. For the linear estimator, one expects that the variance should decrease linearly with the number of samples. For the quadratic estimator, the variance is $\mathcal O(Bd/s^2 + 1/s)$ by \Cref{lem:quadratic_covariance}. Therefore, whenever $Bd/s^2$ dominates $1/s$, the variance should decrease \emph{quadratically} in the number of samples. Since the plots are shown on a log-log scale, this should result in a slope of $-2$. We see this scaling in the graph shown on the right since $d$ is large (the slope of the regression for the linear estimator is $-.998$ and the slope for the quadratic estimator is $-1.936$). However, for the left graph $Bd = d^2$ is only $64$, so we expect that the variance for the quadratic estimator to scale linearly after about $s = 64$ copies. Indeed, one can observe that the lines for the linear and quadratic estimators are essentially parallel after that point. As a final observation, we note that in both graphs the quadratic estimator becomes better than the linear estimator at the point where the variance becomes less than 1. Since the estimate is only useful once the variance is less than 1, one could interpret this as conveying the fact that the quadratic estimator is \emph{always} better than the linear estimator in this particular parameter regime.

\begin{figure}[h]
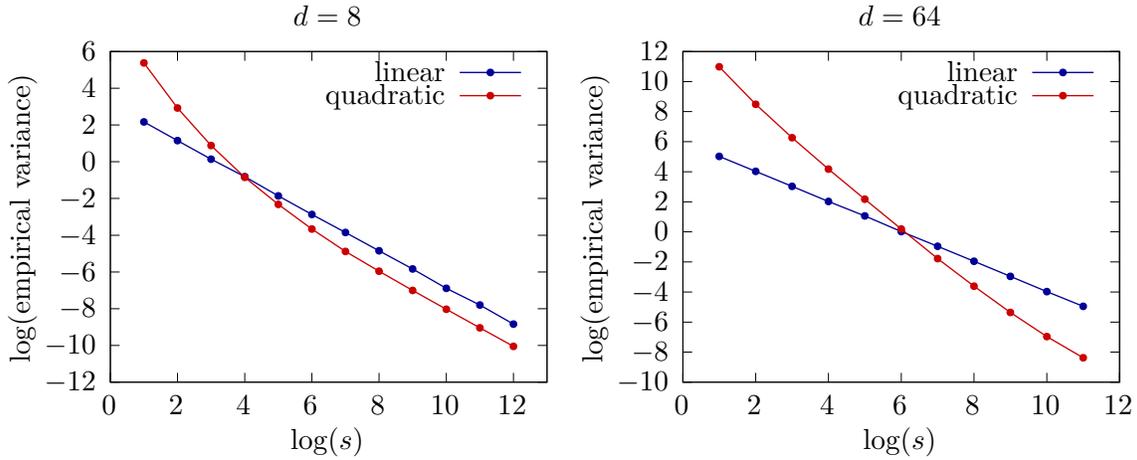

\centering
\begin{minipage}{.48\linewidth}
    \input{figures/Z3.tex}
\end{minipage}
\begin{minipage}{.48\linewidth}
    \input{figures/Z6.tex}
\end{minipage}
\caption{Empirical variances of the linear estimator ($\hat X$, the estimator of \cite{Huang2020}) and our quadratic estimator ($\hat Y$, defined in \Cref{sec:independent_measurements}) for the classical shadows task using independent measurements. The data confirm our variance calculation in \Cref{lem:quadratic_covariance}: $\mathcal O(\frac{Bd}{s^2} + \frac{1}{s})$. Specifically, the quadratic estimator variance scales inverse-quadratically in $s$ when $Bd/s^2$ dominates $1/s$ (a slope on the graph of $-2$), and the linear estimator scales inverse-linearly in $s$ (a slope of $-1$). 
Each point in the graphs represent an empirical variance of $10^4$ trials of the following procedure: generate a random pure state $\rho$ of dimension $d = 2^n$ and a random full-weight Pauli operator $P$ on $n$ qubits; independently measure $s$ copies of $\rho$ in a random basis to obtain outcomes $\psi_1, \ldots, \psi_s$; compute $\rho_i = (d+1) \psi_i - I$ for all $i \in [s]$; compute estimates $x = \frac{1}{s} \sum_{i=1}^{s} \rho_i$ and $y = \frac{1}{s(s-1)} \sum_{i \neq j} \rho_i \rho_j$; output error from true expectation value: $\Tr(P x) - \Tr(P \rho )$ and $\Tr(P y) - \Tr(P \rho)$. See main text for detailed explanation of the slope of the lines. 
}
\label{fig:linear-vs-quadratic-simulation}
\end{figure}

% -------------------------------------------------------------
% SECTION - OPEN PROBLEMS
% -------------------------------------------------------------

\section{Open Problems}
\label{sec:outlook}

Our upper and lower bounds almost completely settle the question of sample complexity for the classical shadows task with arbitrary measurements and arbitrary observables, but clearly many questions remain. First, can the remaining discrepancies between our upper and lower bounds be removed? That is, can we get the lower bounds to have the correct dependence on $\delta$ (which we conjecture to be $\log(\delta^{-1})$) and remove a $\log B$ factor? 

Second, the sample complexity of learning states in the context of tomography is known for all combinations of independent vs.\ joint measurements, and pure vs.\ mixed states. Can we characterize the sample complexity of classical shadows as thoroughly? Ref.~\cite{Huang2020} gives matching bounds for independent measurements and mixed states, and our result gives matching bounds for joint measurements on pure states, but the independent/pure and joint/mixed cases are open. At the very least, we know that in the independent measurement setting that the pure state case is different than the mixed state setting since our upper bound in \Cref{thm:ps_im_ub} is smaller than the lower bound in Ref.~\cite{Huang2020} for some regimes.

As a follow up, tomography bounds are sometimes stated as a function of $r$, the rank of the unknown state. For example, the sample complexity is $\tilde{\Theta}(rd \epsilon^{-1})$ for joint measurements, capturing the pure state $(r=1)$ and worst-case mixed state ($r = d$) behaviour simultaneously. We believe the sample complexity of the classical shadows task depends smoothly on $r$, but do not yet have a conjecture. 

Third, we do not describe how to concretely implement our large joint measurements. We may replace the continuum of Haar random pure states and elements $A_{\psi}$ in the POVM $\mathcal M_s$ with a concrete $(s+2)$-design, but even then it is not clear how to practically implement the measurement. Alternatively, can we make a similar POVM with a simpler ensemble of states, e.g., $t$-designs for some $t< s+2$, and update the analysis to achieve an equivalent end result? We note that for our exact measurement, a lower bound on $t$ can be computed by using upper bounds on the number of bits required to describe a state $t$-design. To see this, notice that the outcome of the POVM $\mathcal M_s$ suffices as a compressed description of the state for the purpose of observable estimation. Since we proved a state compression lower bound of $\Omega(\sqrt{B}\epsilon^{-1})$-many bits, the number of bits required to specify the state must be at least this large. In particular, this implies that you could not implement our measurement on $n$ qubits with a $3$-design since Clifford states can be specified using $\mathcal O(n^2)$-many bits.

Fourth, there is the question of robustness to error or noise. For any classical shadows protocol, we can ask how it behaves when the samples are not exactly of the form $\rho^{\otimes s}$, due to variation in samples. For our pure state protocols, we are also interested in how fast our algorithms degrade when given mixed states that are close to pure. 

Fifth, Ref.~\cite{Huang2020} introduced two classical shadows protocols known as the random Pauli measurements and the random Clifford measurements schemes. The former schemes targets local observable. The latter scheme, like our protocols, targets observables with low Frobenius norm. These target classes of observables are mutually exclusive and each scheme achieves lower sample complexity with respect to its target class observable. Recent work~\cite{bertoni2022shallow,akhtar2022scalable} has also focused on the development of an intermediate scheme that achieves favourable sample complexity scaling in both target classes of observables. All these works focus on general states and independent measurements. Our work identifies an optimal protocol for the low Frobenius norm class of observable in the setting of pure states and joint measurements. So it is natural to consider the pure states and/or joint measurements setting in the context of local observables or the combined class.

Finally, one may consider cubic or higher order generalizations of the quadratic estimator used in the proof of \cref{thm:ps_im_ub}. We leave the analysis of such estimators to future work. 

% -------------------------------------------------------------
% ACKNOWLEDGEMENTS
% -------------------------------------------------------------

\section*{Acknowledgements}
We thank Anurag Anshu, Stephen Bartlett, Daniel Burgarth, David Gosset, David Gross and Richard Kueng for useful discussions. 
Research at Perimeter Institute is supported in part by the Government of Canada through the Department of Innovation, Science and Economic Development Canada and by the Province of Ontario through the Ministry of Colleges and Universities. HP also acknowledges the support of the Natural Sciences and Engineering Research Council of Canada (NSERC) discovery grants [RGPIN-2019-04198] and [RGPIN-2018-05188]. 

\bibliographystyle{quantum}
\bibliography{bibRef}

\begin{thebibliography}{10}

\bibitem{Haah2016sample}
Jeongwan Haah, Aram~W. Harrow, Zhengfeng Ji, Xiaodi Wu, and Nengkun Yu.
\newblock ``Sample-optimal tomography of quantum states''.
\newblock In Proceedings of the Forty-Eighth Annual ACM Symposium on Theory of Computing (STOC 2016).
\newblock \href{https://dx.doi.org/10.1145/2897518.2897585}{Page 913–925}.
\newblock New York, NY, USA~(2016). Association for Computing Machinery.

\bibitem{ODonnel2016efficient}
Ryan O'Donnell and John Wright.
\newblock ``Efficient quantum tomography''.
\newblock In Proceedings of the Forty-Eighth Annual ACM Symposium on Theory of Computing (STOC 2016).
\newblock \href{https://dx.doi.org/10.1145/2897518.2897544}{Page 899–912}.
\newblock New York, NY, USA~(2016). Association for Computing Machinery.

\bibitem{Aaronson2018shadow}
Scott Aaronson.
\newblock ``Shadow tomography of quantum states''.
\newblock In Proceedings of the 50th Annual ACM SIGACT Symposium on Theory of Computing (STOC 2018).
\newblock \href{https://dx.doi.org/10.1145/3188745.3188802}{Page 325–338}.
\newblock New York, NY, USA~(2018). Association for Computing Machinery.

\bibitem{buadescu2021improved}
Costin B\u{a}descu and Ryan O'Donnell.
\newblock ``Improved quantum data analysis''.
\newblock In Proceedings of the 53rd Annual ACM SIGACT Symposium on Theory of Computing (STOC 2021).
\newblock \href{https://dx.doi.org/10.1145/3406325.3451109}{Page 1398–1411}.
\newblock New York, NY, USA~(2021). Association for Computing Machinery.

\bibitem{chen2022exponential}
Sitan Chen, Jordan Cotler, Hsin-Yuan Huang, and Jerry Li.
\newblock ``Exponential separations between learning with and without quantum memory''.
\newblock In 2021 IEEE 62nd Annual Symposium on Foundations of Computer Science (FOCS 2022).
\newblock \href{https://dx.doi.org/10.1109/FOCS52979.2021.00063}{Pages 574--585}.
\newblock Los Alamitos, CA, USA~(2022).

\bibitem{Huang2020}
Hsin-Yuan Huang, Richard Kueng, and John Preskill.
\newblock ``Predicting many properties of a quantum system from very few measurements''.
\newblock \href{https://dx.doi.org/10.1038/s41567-020-0932-7}{Nature Physics {\bf 16}, 1050--1057}~(2020).

\bibitem{gosset_smolin2019}
David Gosset and John Smolin.
\newblock ``A compressed classical description of quantum states''.
\newblock In 14th Conference on the Theory of Quantum Computation, Communication and Cryptography (TQC 2019).
\newblock \href{https://dx.doi.org/10.4230/LIPIcs.TQC.2019.8}{Volume 135 of Leibniz International Proceedings in Informatics (LIPIcs), pages 8:1--8:9}.
\newblock Dagstuhl, Germany~(2019). Schloss Dagstuhl--Leibniz-Zentrum fuer Informatik.

\bibitem{scott2006tight}
A~J Scott.
\newblock ``Tight informationally complete quantum measurements''.
\newblock \href{https://dx.doi.org/10.1088/0305-4470/39/43/009}{Journal of Physics A: Mathematical and General {\bf 39}, 13507}~(2006).

\bibitem{gross2015partial}
D.~Gross, F.~Krahmer, and R.~Kueng.
\newblock ``A partial derandomization of {P}hase{L}ift using spherical designs''.
\newblock \href{https://dx.doi.org/10.1007/s00041-014-9361-2}{Journal of Fourier Analysis and Applications {\bf 21}, 229--266}~(2015).

\bibitem{roberts2017chaos}
Daniel~A. Roberts and Beni Yoshida.
\newblock ``Chaos and complexity by design''.
\newblock \href{https://dx.doi.org/10.1007/JHEP04(2017)121}{Journal of High Energy Physics {\bf 2017}, 121}~(2017).

\bibitem{lugosi2019mean}
G.~Lugosi and S.~Mendelson.
\newblock ``Mean estimation and regression under heavy-tailed distributions: A survey''.
\newblock \href{https://dx.doi.org/10.1007/s10208-019-09427-x}{Found. Comp. Math. {\bf 19}, 1145--1190}~(2019).

\bibitem{lerasle2019lecture}
M.~Lerasle.
\newblock ``Lecture notes: Selected topics on robust statistical learning theory''~(2019).
\newblock  \href{http://arxiv.org/abs/1908.10761}{arXiv:1908.10761}.

\bibitem{bajnok1992construction}
Bela Bajnok.
\newblock ``Construction of spherical $t$-designs''.
\newblock \href{https://dx.doi.org/10.1007/BF00147866}{Geometriae Dedicata {\bf 43}, 167--179}~(1992).

\bibitem{hayashi2005reexamination}
A.~Hayashi, T.~Hashimoto, and M.~Horibe.
\newblock ``Reexamination of optimal quantum state estimation of pure states''.
\newblock \href{https://dx.doi.org/10.1103/PhysRevA.72.032325}{Phys. Rev. A {\bf 72}, 032325}~(2005).

\bibitem{bondarenko2013optimal}
Andriy Bondarenko, Danylo Radchenko, and Maryna Viazovska.
\newblock ``Optimal asymptotic bounds for spherical designs''.
\newblock \href{https://dx.doi.org/10.4007/annals.2013.178.2.2}{Annals of Mathematics {\bf 178}, 443–452}~(2013).

\bibitem{nielsen2010quantum}
Michael~A Nielsen and Isaac~L Chuang.
\newblock ``Quantum computation and quantum information''.
\newblock \href{https://dx.doi.org/10.1017/CBO9780511976667}{Cambridge University Press}. ~(2010).

\bibitem{Webb2016}
Zak Webb.
\newblock ``The {C}lifford group forms a unitary 3-design''.
\newblock \href{https://dx.doi.org/10.26421/qic16.15-16-8}{Quantum Information and Computation {\bf 16}, 1379–1400}~(2016).

\bibitem{Zhu2017}
Huangjun Zhu.
\newblock ``Multiqubit {C}lifford groups are unitary 3-designs''.
\newblock \href{https://dx.doi.org/10.1103/physreva.96.062336}{Physical Review A{\bf 96}}~(2017).

\bibitem{kueng2015qubit}
Richard Kueng and David Gross.
\newblock ``Qubit stabilizer states are complex projective 3-designs''~(2015).
\newblock  \href{http://arxiv.org/abs/1510.02767}{arXiv:1510.02767}.

\bibitem{zhu2016clifford}
Huangjun Zhu, Richard Kueng, Markus Grassl, and David Gross.
\newblock ``The {C}lifford group fails gracefully to be a unitary 4-design''~(2016).
\newblock  \href{http://arxiv.org/abs/1609.08172}{arXiv:1609.08172}.

\bibitem{Gavinsky2007exponential}
Dmitry Gavinsky, Julia Kempe, Iordanis Kerenidis, Ran Raz, and Ronald {De Wolf}.
\newblock ``Exponential separations for one-way quantum communication complexity, with applications to cryptography''.
\newblock In Proceedings of the Annual ACM Symposium on Theory of Computing (STOC 2007).
\newblock \href{https://dx.doi.org/10.1145/1250790.1250866}{Pages 516--525}.
\newblock New York, NY, USA~(2007). Association for Computing Machinery.

\bibitem{helstrom1976quantum}
Carl~W. Helstrom.
\newblock ``Quantum detection and estimation theory''.
\newblock \href{https://dx.doi.org/10.1007/BF01007479}{Journal of Statistical Physics {\bf 1}, 231--252}~(1969).

\bibitem{Fuchs1999}
C.A. Fuchs and J.~van~de Graaf.
\newblock ``Cryptographic distinguishability measures for quantum-mechanical states''.
\newblock \href{https://dx.doi.org/10.1109/18.761271}{IEEE Transactions on Information Theory {\bf 45}, 1216--1227}~(1999).

\bibitem{yao1977probabilistic}
Andrew Chi-Chin Yao.
\newblock ``Probabilistic computations: Toward a unified measure of complexity''.
\newblock In 18th Annual Symposium on Foundations of Computer Science (SFCS 1977).
\newblock \href{https://dx.doi.org/10.1109/SFCS.1977.24}{Pages 222--227}.
\newblock IEEE~(1977).

\bibitem{harsha2007communication}
Prahladh Harsha, Rahul Jain, David McAllester, and Jaikumar Radhakrishnan.
\newblock ``The communication complexity of correlation''.
\newblock \href{https://dx.doi.org/10.1109/TIT.2009.2034824}{IEEE Transactions on Information Theory {\bf 56}, 438--449}~(2010).

\bibitem{holevo1973bounds}
Alexander~Semenovich Holevo.
\newblock ``Bounds for the quantity of information transmitted by a quantum communication channel''.
\newblock Problemy Peredachi Informatsii {\bf 9}, 3--11~(1973).
\newblock  url:~\href{http://mi.mathnet.ru/ppi903}{http://mi.mathnet.ru/ppi903}.

\bibitem{bertoni2022shallow}
Christian Bertoni, Jonas Haferkamp, Marcel Hinsche, Marios Ioannou, Jens Eisert, and Hakop Pashayan.
\newblock ``Shallow shadows: Expectation estimation using low-depth random {Clifford} circuits''~(2022).
\newblock  \href{http://arxiv.org/abs/2209.12924}{arXiv:2209.12924}.

\bibitem{akhtar2022scalable}
Ahmed~A. Akhtar, Hong-Ye Hu, and Yi-Zhuang You.
\newblock ``Scalable and flexible classical shadow tomography with tensor networks''.
\newblock \href{https://dx.doi.org/10.22331/q-2023-06-01-1026}{{Quantum} {\bf 7}, 1026}~(2023).

\end{thebibliography}

\appendix

% -------------------------------------------------------------
% APPENDIX - PROPER LEARNING DISCUSSION
% -------------------------------------------------------------

\section{Proper Learning Discussion}
\label{sec:proper_learning}

The classical shadows task is a learning problem: given many samples of a quantum state, we make measurements with the goal of learning the state well enough to approximate arbitrary observables. A learning problem is said to be \emph{proper} if it requires the learned representation to be from the same class---in our case, the class of pure states---as the original object. Our classical shadows algorithm fails to be proper on several counts:
\begin{enumerate}
    \item The output of the problem is a real number, not the classical description of a quantum state. By definition, the classical shadows task is not a proper learning task.
    \item Internally, our algorithm does produce Hermitian matrices $\rhohat^{(i)}$ with trace $1$ (the \emph{shadows}), which have the potential to represent a quantum state. However, each such estimate is
    \begin{enumerate}
    \item high rank, so it cannot represent a pure state, and 
    \item not positive semi-definite, so it cannot represent a mixed state.
    \end{enumerate}
    \item The algorithm uses multiple $\rhohat^{(i)}$ shadows, and takes a median of their estimates on each observable. Even if each $\rhohat^{(i)}$ were a quantum state, there may not exist a state exactly consistent with all our medians on a set of observables. 
\end{enumerate}

On the other hand, in the limit where the failure probability is extremely small, we can afford to estimate the expectation of the state on an $\epsilon$-cover of the observables. From these values, we can approximate the original state to accuracy $\epsilon$ in trace distance. In this regime, the problem is equivalent to tomography, which \emph{is} a proper learning problem.

In this section, we show that any proper learning algorithm for classical shadows would require significantly more samples. Our starting point is a known lower bound for the quantum state tomography question for pure states:
\begin{theorem}[\cite{Haah2016sample}]
\label{thm:tomography_lower_bound}
Any quantum algorithm that takes copies of an unknown pure state $\rho$ and outputs a classical estimate $\rhohat$ such that $\pnorm{\rho - \rhohat}{1} \le \epsilon$ with constant failure probability $\delta < 1$ requires $\Omega(d \epsilon^{-2} / \log(d/\epsilon))$ samples.
\end{theorem}
In particular, we show that a proper classical shadows algorithm implies a state tomography algorithm.

\begin{thm}\label{thm:proper_learning_LB}
Suppose there exists a quantum learning algorithm that, given $s$ copies of an unknown $d$-dimensional state $\rho$, outputs a classical description of a trace $1$, Hermitian PSD matrix $\hat{\rho}$ such that, for all $O\in \obs(1)$ with failure probability $\delta < 1$:
\begin{align}
    \abs{\tr{O\rho}-\tr{O\hat{\rho}}}\leq \epsilon.
\end{align}
Then, $s = \tilde\Omega(d/\epsilon)$. 
\end{thm}

\begin{proof}
Run the algorithm and feed it $O = \rho$. We have that $\pnorm{\rho-\hat{\rho}}{1}\leq \sqrt{8\epsilon}$ since
\begin{align*}
    \epsilon &\geq |  \Tr(O \rho) - \Tr(O \rhohat) | \\
    &= |\Tr( \rho^2 )- \Tr( \rho \rhohat )| \\
    &= |1 - \Tr( \rho \rhohat )| & \text{$\rho = \rho^2$ since $\rho$ is pure} \\
    &= |1 - F( \rho,  \rhohat )^2| & \text{$\Tr(\rho \sigma) = F(\rho, \sigma)^2$ if either is pure} \\
    &\geq 1 - F(\rho, \rhohat)^2 \\
    &\geq \tfrac{1}{8} \| \rho - \rhohat \|_1^{2} & \text{Fuchs-van de Graaf inequality}
\end{align*}
where $F(\rho, \sigma) := \Tr(\sqrt{\sqrt{\rho} \sigma \sqrt{\rho}})$ is the \emph{fidelity} of $\rho$ and $\sigma$. 

Hence, if we solve the classical shadows task to error $\epsilon$ on this observable, then we have estimated $\rho$ to within Schatten $1$-norm distance $\mathcal{O}(\sqrt{\epsilon})$. That is, classical shadows learner is also a quantum tomography algorithm, so we can use the known lower bound from \Cref{thm:tomography_lower_bound}. It follows that a \emph{proper} learning algorithm for the classical shadows task requires $\tilde\Omega(d/\epsilon)$ samples.  
\end{proof}
In the commonly considered regime where $d\gg\epsilon^{-1}$, the above lower bound is significantly more than both our algorithm (from Section~\ref{sec:upper_bound}) and the original classical shadows algorithm, which use only  $\bigo{\tfrac{1}{\epsilon^2}}$ samples.

% -------------------------------------------------------------
% APPENDIX - COVARIANCE BOUNDS
% -------------------------------------------------------------

\section{Covariance bounds}
\label{subsec:covariance_bounds}
The goal of this subsection is to prove \Cref{lem:all_the_covariances}, which gives each of the covariance terms $\Cov(\Tr(O \rhohat_i \rhohat_j), \Tr(O \rhohat_k \rhohat_{\ell}))$ that appears in the expansion of the estimator $\hat Y$. Because $\rhohat_j$ appears twice in all of the covariance terms that are non-zero, it will be convenient to explicitly calculate the second moment of $\rhohat_j$. 
\begin{lemma}
\label{lem:second_rhohat}
For all $j$, the second moment of $\rhohat_j$ is 
$$
\E[\rhohat_j^{\otimes 2}] = \left( I \otimes I + I \otimes \rho + \rho \otimes I \right) \left( W_{(1\,2)} - \frac{2}{d+2}\sym^{(2)} \right).
$$
\end{lemma}
\begin{proof}
Recall that $\rhohat_j$ is obtained through an independent and identical measurement process, so it suffices to analyze a specific $\rhohat := (d+1)\Psi - I$ term. By \Cref{lem:first_moment} and \Cref{lem:second_moment}, we compute the first and second moment of $\Psi$ for the special case $s = 1$ as 
$$
\E[\Psi] = \frac{I + \rho}{d+1} \; \text{ and } \;        
\E[\Psi \otimes \Psi] = \frac{I \otimes I + I \otimes \rho + \rho \otimes I}{(d+1)(d+2)} \left( W_{(1)(2)} + W_{(1\,2)} \right) 
$$
Now, expanding out the second moment, we get
\begin{align*}
    \E[\rhohat^{\otimes 2}] &= \E[ ((d+1)\Psi - I)^{\otimes 2}] \\
    &= (d+1)^2 \E[ \Psi \otimes \Psi] - (d+1)(\E[\Psi] \otimes I + I \otimes \E[\Psi]) + I \otimes I \\
    &= \frac{d+1}{d+2} \left(I \otimes I + I \otimes \rho + \rho \otimes I \right)(W_{(1)(2)} + W_{(1\,2)}) - \left(I \otimes I + I \otimes \rho + \rho \otimes I \right) \\
    &= \left(I \otimes I + I \otimes \rho + \rho \otimes I \right) \left( \frac{(d+1) W_{(1\,2)} - W_{(1)(2)}}{d+2} \right),
\end{align*}
where we recall that $W_{(1)(2)} = I \otimes I$ to obtain the last equality. We arrive at the lemma by writing the final line in terms of the symmetric subspace $\sym^{(2)} = (W_{(1\,2)} + W_{(1)(2)})/2$.
\end{proof}
Our goal will now be to express all covariance terms in a manner such that we can apply \Cref{lem:second_rhohat}. Since these equations can become quite cumbersome to write out fully, we will often drop the ``$\otimes$'' symbol in expressions with $I$ and $\rho$. For example, 
$$
I \otimes I + \rho \otimes I + I \otimes \rho \to (II + I \rho + \rho I)
$$
will be a common abbreviation. We enclose these abbreviations in parentheses when they are multiplied with other terms.

Let's first tackle the covariance terms $\Cov(\Tr(O \rhohat_i \rhohat_j), \Tr(O \rhohat_k \rhohat_{\ell}))$ where there is only 1 index shared, i.e., $(|\{ i, j \} \cap \{ k, \ell \}| = 1)$. There are two subcases: a match in different positions ($i = \ell$ or $j = k$); or a match in same position ($i = k$ or $j = \ell$). While the proofs are quite similar, we break them into two separate corollaries.

\begin{cor}
    \label{cor:ijjk}
    For all distinct $i, j, k$,
    $$\Cov( \Tr( O\rhohat_i \rhohat_j ), \Tr( O\rhohat_j \rhohat_k ) ) \leq 2\Tr(O \rho)^2 \leq 2\| O \|^2 .$$
\end{cor}
\begin{proof}
    First, translate covariance to a second moment calculation: 
    \begin{align*}
        \Cov( \Tr( O\rhohat_i \rhohat_j ), \Tr( O\rhohat_j \rhohat_k ) )&= \E[\Tr( O\rhohat_i \rhohat_j ) \Tr( O\rhohat_j \rhohat_k )^{*}] - \Tr(O \rho) \Tr(O \rho)^{*} \\
        &= \E[\Tr( O\rhohat_i \rhohat_j ) \Tr( O\rhohat_k \rhohat_j )] - \Tr(O \rho)^2,
    \end{align*}
    where the last equality uses that $O$, $\rho_j$, $\rho_k$ are Hermitian. The second moment can be further decomposed using independence of $\rhohat_i$, $\rhohat_j$, and $\rhohat_k$.  
    \begin{align*}
    \E[\Tr( O\rhohat_i \rhohat_j ) \Tr( O\rhohat_k \rhohat_j )]
    &= \Tr \left( (O \otimes O) \, \E[\rhohat_i \rhohat_j \otimes \rhohat_k \rhohat_j] \right) \\
    &= \Tr \left( (O \otimes O) \, (\E[\rhohat_i] \otimes \E[\rhohat_k]) \, \E[\rhohat_j \otimes \rhohat_j] \right)
    \end{align*}
    Using $\E[\rhohat_i] = \E[\rhohat_k] = \rho$ and \Cref{lem:second_rhohat}, we have
    \begin{align*}
    (\E[\rhohat_i] \otimes \E[\rhohat_k]) \E[\rhohat_j \otimes \rhohat_j] 
    &= (\rho \rho) \left( II + \rho I + I \rho \right)( W_{(1\,2)} - 2\sym^{(2)}/(d+2) ) \\
    &= 3 (\rho \rho) ( W_{(1\,2)} - 2\sym^{(2)}/(d+2) ),
    \end{align*}
    where we have once again use the purity of $\rho$. Plugging everything in, we get
    \begin{align*}
    \Cov( \Tr( O\rhohat_i \rhohat_j ), \Tr( O\rhohat_j \rhohat_k ) ) &= 3\Tr( O^{\otimes 2} \rho^{\otimes 2} W_{(1\,2)}) - \frac{6}{d+2} \Tr( O^{\otimes 2} \rho^{\otimes 2} \sym^{(2)} ) - \Tr(O \rho)^2 \\
    &= 3 \Tr(O \rho)^2 - \frac{6}{d+2} \Tr(O \rho)^2 - \Tr(O \rho)^2 \\
    &\leq 2 \Tr(O \rho)^2 \leq 2 \| O \|^2. 
    \end{align*}
\end{proof}
\begin{cor}
    \label{cor:ijkj}
    For all distinct $i, j, k$,
    $$\Cov( \Tr( O\rhohat_i \rhohat_j ), \Tr( O\rhohat_k \rhohat_j ) ) \leq 2 \Tr(O^2 \rho) \leq 2 \| O \|^2.$$
\end{cor}
\begin{proof}
The proof is similar to that of \Cref{cor:ijjk}. We expand the covariance as 
$$
\Cov(\Tr(O \rhohat_i \rhohat_j), \Tr(O \rhohat_k \rhohat_j)) = \E[\Tr(O \rhohat_i \rhohat_j) \Tr(O \rhohat_j \rhohat_k)] - \Tr(O \rho)^2
$$
and compute the second moment term using independence:
\begin{align*}
    \E[\Tr(O \rhohat_i \rhohat_j) \Tr(O \rhohat_j \rhohat_k)]
    &= \Tr \big((O \otimes O) (\E[\rhohat_i] \otimes I) \E[ \rhohat_j \otimes \rhohat_j] (I \otimes \E[\rhohat_k])] \big) \\
    &= \Tr ( (O \otimes O) (\rho I)  ( II + \rho I + I\rho ) ( W_{(1\,2)} - 2\sym^{(2)}/(d+2) ) (I \rho) )
\end{align*}
For the $W_{(1 \, 2)}$ term, we get 
$$
\Tr( O^{\otimes 2} (\rho I)(II + I\rho + \rho I) W_{(1\,2)} (I \rho)] = \Tr( O^{\otimes 2} (2 \rho I + \rho \rho) W_{(1\,2)} ) = 2 \Tr(O^2 \rho) + \Tr(O \rho)^2.  
$$
The $\sym^{(2)}$ term subtracts a positive quantity, so we drop it to get an upper bound on covariance.
$$
\Cov(\Tr(O \rhohat_i \rhohat_j), \Tr(O \rhohat_k \rhohat_j)) \leq 2 \Tr(O^2 \rho) + \Tr(O \rho)^2 - \Tr(O \rho)^2 = 2 \Tr(O^2 \rho) \leq 2 \| O \|^2.
$$
\end{proof}

We now turn to the covariance terms $\Cov(\Tr(O \rhohat_i \rhohat_j), \Tr(O \rhohat_k \rhohat_{\ell}))$ which share two indices, i.e., $(|\{ i, j \} \cap \{ k, \ell \}| = 2)$. Once again, there are two subcases: the order is swapped ($i = \ell$ and $j = k$); or the order is the same ($i = j$ and $k = \ell$). 

In both cases, their are terms of the covariance that are proportional to $\Tr(O)$, but interestingly, we cannot assume $O$ is traceless as we have earlier. This is due to the fact that the $\hat Y$ estimator does not necessarily have trace 1. Nevertheless, it will turn out that this cannot affect the overall covariance of $\hat Y$ too much, as we will show in the following corollaries.

\begin{cor}
\label{cor:ijji}
    For all distinct $i, j$,
    $$\Cov( \Tr( O\rhohat_i \rhohat_j ), \Tr( O\rhohat_j \rhohat_i ) ) \leq d \Tr(O^2) + 6 \sqrt{d \Tr(O^2)} + \| O \|^2$$
\end{cor}
\begin{proof}
Using independence, we expand the covariance as 
$$
\Cov( \Tr( O\rhohat_i \rhohat_j ), \Tr( O\rhohat_j \rhohat_i ) ) = \Tr( (O \otimes O) \E[ \rhohat_i \otimes \rhohat_i] \E[ \rhohat_j \otimes \rhohat_j] ) - \Tr(O \rho)^2.
$$
Using \Cref{lem:second_rhohat}, we get an expression for the second moment terms:
\begin{align*}
\E[ \rhohat_i \otimes \rhohat_i] \E[ \rhohat_j \otimes \rhohat_j] 
&= ( II + \rho I + I \rho)^2 ( W_{(1\,2)} - 2\sym^{(2)}/(d+2) )^2 \\
&=\left( II + 3 I \rho + 3 \rho I + 2 \rho \rho \right) ( W_{(1)(2)} - 4\sym^{(2)}/(d+2) + 4\sym^{(2)}/(d+2)^2 ) \\
&=\left( II + 3 I \rho + 3 \rho I + 2 \rho \rho \right) ( W_{(1)(2)} - 4(d+1)\sym^{(2)}/(d+2)^2).
\end{align*}
For the $W_{(1)(2)}$ term in $\Tr( (O \otimes O) \E[ \rhohat_i \otimes \rhohat_i] \E[ \rhohat_j \otimes \rhohat_j])$, we get
$$
\Tr( O^{\otimes 2} (II + 3 I \rho + 3 \rho I + 2 \rho \rho) W_{(1)(2)}) = \Tr(O)^2 + 6\Tr(O) + 2 \Tr(O \rho)^2.
$$
For the $\sym^{(2)}$ term, we get (ignoring the scalar factor)
$$
\Tr( O^{\otimes 2} (II + 3 I \rho + 3 \rho I + 2 \rho \rho) \sym^{(2)}) = \frac{\Tr(O^2) + \Tr(O)^2 + 6 \Tr(O) + 6\Tr(O^2 \rho) + 4 \Tr(O\rho)^2}{2}.
$$
Therefore, the covariance $\Cov( \Tr( O\rhohat_i \rhohat_j ), \Tr( O\rhohat_j \rhohat_i ) )$ is 
\begin{align*}
    &= \Tr(O)^2 + 6\Tr(O) + \Tr(O \rho)^2 - \frac{2(d+1)}{(d+2)^2} \left( \Tr(O^2) + \Tr(O)^2 + 6 \Tr(O) + 6\Tr(O^2 \rho) + 4 \Tr(O\rho)^2 \right) \\
    &\le \Tr(O)^2 + 6|\Tr(O)| + \Tr(O \rho)^2 \\
    &\le d \Tr(O^2) + 6 \sqrt{d \Tr(O^2)} + \| O \|^2
\end{align*}
where the first inequality drops the negative terms and the last line comes from Cauchy-Schwarz.
\end{proof}

\begin{cor}
\label{cor:ijij}
    For all distinct $i, j$,
    $$\Cov( \Tr( O\rhohat_i \rhohat_j ), \Tr( O\rhohat_i \rhohat_j ) ) = (d+2) \Tr(O^2) + (3d-2) \| O \|^2.$$
\end{cor}
\begin{proof}
    We expand the covariance as usual and obtain the following.
    $$
    \Cov( \Tr( O\rhohat_i \rhohat_j ), \Tr( O\rhohat_i \rhohat_j ) ) = \Tr( (O \otimes O) \E[ \rhohat_i \rhohat_j \otimes \rhohat_j \rhohat_i] ) - \Tr(O \rho)^2
    $$
    It is not so easy to decompose this into $\E[ \rhohat_i \otimes \rhohat_i ]$ and $\E[ \rhohat_j \otimes \rhohat_j]$. We introduce a third qudit, and use the following identity:
    $$
    \rhohat_i \rhohat_j \otimes \rhohat_j \rhohat_i = \Tr_{3}( (I \otimes \rhohat_i \otimes \rhohat_i) (\rhohat_j \otimes \rhohat_j \otimes I) W_{(1\,3)(2)} ).
    $$
    So, plugging back into the expectation, we get 
    \begin{align*}
        \Tr( (O \otimes O) \E[ \rhohat_i \rhohat_j \otimes \rhohat_j \rhohat_i] ) &=  \Tr( (O \otimes O \otimes I) (I \otimes \E[\rhohat_i \otimes \rhohat_i]) (\E[\rhohat_j \otimes \rhohat_j] \otimes I) W_{(1\,3)(2)} ).
    \end{align*}
    Once again, we use \Cref{lem:second_rhohat} to compute
    \begin{align*}
        I \otimes \E[\rhohat_i \otimes \rhohat_i] &= (III + I\rho I + II\rho)
         \left( \tfrac{d+1}{d+2} W_{(1)(2\,3)} - \tfrac{1}{d+2} W_{(1)(2)(3)} \right), \\
        \E[\rhohat_j \otimes \rhohat_j] \otimes I &= (III + \rho II + I\rho I)\left( \tfrac{d+1}{d+2} W_{(1\,2)(3)} - \tfrac{1}{d+2} W_{(1)(2)(3)} \right).
    \end{align*}
    Each expectation is a difference of two $W_{\pi}$ permutation terms. Therefore, computing the product of the two expectations, we get four $W_{\pi}$ terms. It will turn out that the dominant one is the following, where we have taken the $W_{(1)(2\,3)}$ term for $\rhohat_i$ and the $W_{(1\,2)(3)}$ term from $\rhohat_j$ (to visualize the largest contribution from this term, we refer to tensor network picture in \Cref{fig:dominant_ijij}). Dropping the scalar factor $(d+1)^2/(d+2)^2$ for now, we get
    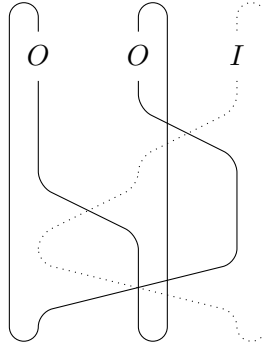
\begin{figure}
        \centering
        \def\rowspacing{1em}
\def\colspacing{2em}
\def\roundness{0.5em}
\def\overdist{1em}
\[
\begin{tikzpicture}[baseline=(current bounding box.center)]
	\matrix (m) [matrix of math nodes, row sep=\rowspacing, column sep=\colspacing, text height=0.8em,text depth=.2em] {
		O & O & I \\
		\phantom{\psi} & \phantom{\psi} & \phantom{\psi} \\
		\phantom{\psi} & \phantom{\psi} & \phantom{\psi} \\
		\phantom{\psi} & \phantom{\psi} & \phantom{\psi} \\
	};
	\draw[rounded corners=\roundness] (m-1-1) -- (m-3-1.north) -- (m-3-2.south) -- (m-4-2.south) -- ++(0,-\overdist) -- ++(\loopwidth,0) |- ($(m-1-2.north) + (0,\overdist)$) -- (m-1-2);
	
	\draw[rounded corners=\roundness] (m-1-2) -- (m-2-2.north) -- (m-2-3.south) -- (m-4-3.north) -- (m-4-1.south) -- ++(0,-\overdist) -- ++(-\loopwidth,0) |- ($(m-1-1.north) + (0,\overdist)$) -- (m-1-1);
	
	\draw[rounded corners=\roundness,dotted] (m-1-3) -- (m-2-3.north) -- (m-2-2.south) --
	(m-3-2.north) -- (m-3-1.south) -- (m-4-1.north) -- (m-4-3.south) -- ++(0,-\overdist) -- ++(\loopwidth,0) |- ($(m-1-3.north) + (0,\overdist)$) -- (m-1-3);
\end{tikzpicture}
\]
        \caption{The dominant term in Corollary~\ref{cor:ijij}, $\Tr( (O \otimes O \otimes I) \, III \, W_{(1)(2\,3)} \, III \, W_{(1\,2)(3)} W_{(1\,3)(2)}) = d \Tr(O^2)$.}
        \label{fig:dominant_ijij}
    \end{figure}
    \begin{align*}
        &\quad \Tr( (O \otimes O \otimes I) (III + I\rho I + II\rho) W_{(1)(2\,3)} (III + \rho II + I\rho I) W_{(1\,2)(3)} W_{(1\,3)(2)} ) \\
        &= \Tr( (O \otimes O \otimes I) (III + I\rho I + II\rho)(III + \rho II + II \rho) W_{(1)(2\,3)}  W_{(1\,2)(3)} W_{(1\,3)(2)} ) \\
        &= \Tr( (O \otimes O \otimes I) (III + 3II \rho + I \rho I + \rho II + \rho I \rho + I \rho \rho + \rho \rho I) W_{(1\,2)(3)} ) \\
        &= \Tr(O^2)(\Tr(I)+3\Tr(\rho)) + \Tr(O^2 \rho)(2\Tr(I) + 2\Tr(\rho)) + \Tr(O \rho)^2 \Tr(I) \\
        &=(d+3) \Tr(O^2) + 2(d+1)\Tr(O^2 \rho) + d \Tr(O \rho)^2.
    \end{align*}
    For completeness, let's also compute the other 3 terms (also without their scalar factors): 
    \begin{align*}
        &\quad \Tr( (O \otimes O \otimes I) (III + I\rho I + II\rho) W_{(1)(2)(3)} (III + \rho II + I\rho I) W_{(1)(2)(3)} W_{(1\,3)(2)} ) \\
        &= \Tr( (O \otimes O \otimes I) (III + \rho II + 3I\rho I + II\rho + \rho I \rho + \rho \rho I + I\rho \rho) W_{(1\,3)(2)} ) \\
        &= \Tr(O)^2 + 6 \Tr(O) \Tr(O \rho) + 2 \Tr(O \rho)^2 \\
        &\le d \Tr(O^2) + 6 \sqrt{d}\Tr(O^2) + 2 \Tr(O \rho)^2 \le 6 d \Tr(O^2) + 2 \Tr(O \rho)^2,
    \end{align*}
    where inequality comes from Cauchy-Schwarz and the fact that $d \ge 2$. The final two terms turn out to be equal:
    \begin{align*}
        &\quad \Tr( (O \otimes O \otimes I) (III + I\rho I + II\rho) W_{(1)(2\,3)} (III + \rho II + I\rho I) W_{(1)(2)(3)} W_{(1\,3)(2)} ) \\
        % &= \Tr( (O \otimes O \otimes I) (III + I\rho I + II\rho)(III + \rho II + II \rho) W_{(1)(2\,3)}  W_{(1)(2)(3)} W_{(1\,3)(2)} ) \\
        &= \Tr( (O \otimes O \otimes I) (III + 3II \rho + I \rho I + \rho II + I \rho \rho  + \rho I \rho + \rho \rho I) W_{(1\,2\,3)} ) \\
        &= \Tr(O^2) + 6 \Tr(O^2\rho) + 2 \Tr(O\rho)^2
    \end{align*}
    and
    \begin{align*}
        &\quad \Tr( (O \otimes O \otimes I) (III + I\rho I + II\rho) W_{(1)(2)(3)} (III + \rho II + I\rho I) W_{(1)(2\,3)} W_{(1\,3)(2)} ) \\
        &= \Tr( (O \otimes O \otimes I)  (III + II\rho + 3I\rho I + \rho II + I\rho \rho+ \rho I \rho + \rho \rho I ) W_{(1\,2\,3)} ) \\
        &= \Tr(O^2) + 6 \Tr(O^2\rho) + 2 \Tr(O\rho)^2
    \end{align*}
    
Notice that these last two terms are non-negative, and so multiplying them by $-(d+1)/(d+2)^2$ makes them non-positive. Since we want to give an upper bound on the covariance, these terms can be dropped. Altogether, and inserting the appropriate constants, we get the following upper bound on the covariance $\Cov( \Tr( O\rhohat_i \rhohat_j ), \Tr( O\rhohat_i \rhohat_j ) )$:
\begin{align*}
    =\,& \Tr( (O \otimes O \otimes I) (I \otimes \E[\rhohat_i \otimes \rhohat_i]) (\E[\rhohat_j \otimes \rhohat_j] \otimes I) W_{(1\,3)(2)} ) \\
    \leq\,& \tfrac{(d+1)^2}{(d+2)^2} ((d+3) \Tr(O^2) + 2(d+1)\Tr(O^2 \rho) + d \Tr(O \rho)^2) + \frac{(6 d \Tr(O^2) + 2 \Tr(O \rho)^2 )}{(d+2)^2} - \Tr(O \rho)^2 \\
    \leq\,& \frac{1}{(d+2)^2} ( (d^3 + 5d^2 + 13 d + 3) \Tr(O^2) + (3 d^3 + 7 d^2 + 3 d) \| O \|^2 ) \\
    \leq\,& (d+2) \Tr(O^2) + (3d-2) \| O \|^2
\end{align*}
where we've used once again that $d \ge 2$.
\end{proof}

\begin{algorithm}[H]
\caption{Algorithm for Theorem~\ref{thm:ps_im_ub}}\label{alg:ps_im_meas}
\begin{algorithmic}[1]
\algrenewcommand\algorithmicrequire{\textbf{Input:}}
\algrenewcommand\algorithmicensure{\textbf{Output:}}
\Require Quantum state $\rho^{\otimes N}$, $B$, $\epsilon$, $\delta$, $d$.
\Ensure Classical shadow $\set{\hat{\rho}^{(i)}}_{i \in [k]}$.
\Statex
\State{$p \gets \bigo{\log(1/\delta)}$} \Comment{Number of batches}
\State{$s \gets \mathcal O(\sqrt{Bd}/\epsilon + 1/\epsilon^2)$} \label{line:sets} \Comment{Samples per batch}

\State{$N \gets p s$} \Comment{Total number of samples}
\For{each batch $i = 1, \ldots, p$}
    \For{$j = 1, \ldots, s$}
    \State{$\psi_j^{(i)} \gets $ Measure fresh $\rho$ with $\mathcal{M}_1$}
    \State{$\rhohat_j^{(i)} \gets (d+1) \psi_j^{(i)} - I$}
    \EndFor
    \State{$\rhohat^{(i)} \gets \frac{1}{s(s-1)} \sum_{j \neq k} \rhohat_j^{(i)} \rhohat_k^{(i)}$} \label{state:estimator}
\EndFor
\State{\Return $\set{\hat{\rho}^{(i)}}_{i \in [k]}$}
\end{algorithmic}
\end{algorithm}

\end{document}